\documentclass[a4paper, envcountsame, envcountsect]{llncs2e/llncs}

\usepackage{graphicx}
\usepackage{enumerate}
\usepackage{enumitem}
\usepackage{graphicx}
\usepackage{algorithm}
\usepackage{algpseudocode}
\usepackage{verbatim}
\usepackage{amsmath}
\usepackage{bbm}
\usepackage{cite}
\usepackage{hyperref}
\usepackage{mathrsfs}
\usepackage{tocloft}
\usepackage{afterpage}
\usepackage{mathtools}

\newcommand{\polylog}{\mathrm{polylog }}

\newlist{inlinelistRoman}{enumerate*}{1}
\setlist*[inlinelistRoman,1]{%
  label=(\Roman*),
}

\newlist{inlinelistroman}{enumerate*}{1}
\setlist*[inlinelistroman,1]{%
  label=(\roman*),
}

\newcommand\Algphase[1]{%
\vspace*{-.7\baselineskip}\Statex\hspace*{\dimexpr-\algorithmicindent-2pt\relax}\rule{\textwidth}{0.4pt}%
\Statex\hspace*{-\algorithmicindent}\textbf{#1}%
\vspace*{-.7\baselineskip}\Statex\hspace*{\dimexpr-\algorithmicindent-2pt\relax}\rule{\textwidth}{0.4pt}%
}

\newcommand{\LINEFORALL}[2]{%
	\State\algorithmicforall\ {#1}\ \algorithmicdo\ {#2} \algorithmicend\ \algorithmicfor%
}

\newenvironment{theorem-repeat}[1]{\begin{trivlist}
\item[\hspace{\labelsep}{\bf\noindent Theorem \ref{#1} }]\em }%
{\end{trivlist}}

\newenvironment{lemma-repeat}[1]{\begin{trivlist}
\item[\hspace{\labelsep}{\bf\noindent Lemma \ref{#1} }]\em }%
{\end{trivlist}}

\def\squareforqed{\hbox{\rlap{$\sqcap$}$\sqcup$}}
\def\qed{\ifmmode\squareforqed\else{\unskip\nobreak\hfil
		\penalty50\hskip1em\null\nobreak\hfil\squareforqed
		\parfillskip=0pt\finalhyphendemerits=0\endgraf}\fi}

\begin{document}
\frontmatter          
\pagestyle{headings}  
\mainmatter              
\title{On Fast and Robust Information Spreading\\ in the Vertex-Congest Model}

\author{Keren Censor-Hillel %
\and Tariq Toukan}
\institute{Technion - Israel Institute of Technology\\
\email{\{ckeren, ttoukan\}@cs.technion.ac.il}
}

\maketitle              

\begin{abstract}
This paper initiates the study of the impact of failures on the fundamental problem of \emph{information spreading} in the Vertex-Congest model, in which in every round, each of the $n$ nodes sends the same $O(\log{n})$-bit message to all of its neighbors.

\quad Our contribution to coping with failures is twofold. First, we prove that the randomized algorithm which chooses uniformly at random the next message to forward is slow, requiring $\Omega(n/\sqrt{k})$ rounds on some graphs, which we denote by $G_{n,k}$, where $k$ is the vertex-connectivity.

\quad Second, we design a randomized algorithm that makes dynamic message choices, with probabilities that change over the execution. We prove that for $G_{n,k}$ it requires only a near-optimal number of $O(n\log^3{n}/k)$ rounds, despite a rate of $q=O(k/n\log^3{n})$ failures per round. Our technique of choosing probabilities that change according to the execution is of independent interest.

\end{abstract}

\section{Introduction}
\label{sec:intro}
Coping with failures is a cornerstone challenge in the design of distributed algorithms. It is desirable that a distributed system continues to operate correctly despite a reasonable amount of failures, and hence obtaining fault-tolerance has been a fundamental goal in this field. The impact of failures has been studied in various models of computation and for various distributed tasks.

In this paper, we initiate the study of robustness against failures of the task of information spreading in the Vertex-Congest model of computation.
Information spreading requires each node of the network to obtain the information of all other nodes. This problem is at the heart of many distributed applications which perform global tasks, and thus is a central issue in distributed computing (see, e.g.,~\cite{peleg00}).
The Vertex-Congest model, where in each round, every node generates an $O(\log{n})$-sized packet and sends it to \emph{all} of its neighbours, abstracts the behavior of wireless networks that operate on top of an abstract MAC layer~\cite{KLN11} that takes care of collisions.

The time required for achieving information spreading depends on the structure of the communication graph. Even without faults, it is clear that 
having a minimum vertex-cut of size $k$ implies an $\Omega(n/k)$ lower bound for the running time of any algorithm in the above model, and hence
our study addresses the $k$-vertex-connectivity of the graph.
The diameter of a graph is a trivial lower bound on the number of rounds required for spreading even without faults, and hence, for $k$-vertex-connected graphs, $\Omega(n/k)$ is a general lower bound as there exist $k$-vertex-connected graphs of diameter $n/k$.

A tempting approach would be to use randomization for choosing which message to forward in each round of communication, in the hope that this would be naturally robust against failures. However, we show that the uniform randomized algorithm is slow on a $k$-vertex-connected family of graphs, denoted $G_{n,k}$, which consists of $n/k$ cliques of size $k$ that are connected by perfect matchings, requiring $\Omega(n/\sqrt{k})$ rounds.

Instead, this paper presents an algorithm for spreading information in the Vertex-Congest model that uses dynamic probabilities for selecting the messages to be sent in each round. We prove that for $G_{n,k}$, the round complexity of our algorithm is almost optimal and that it is highly robust against node failures.

\vspace{-10pt}
\subsection{Our Contribution}
As explained, our first contribution is proving that the intuitive idea of simply choosing at random which message to forward is not efficient.
The proof is based on the fact that there is an inverse proportion between the number of received messages in a node and the probability of a message in that node to be chosen and forwarded.
The larger the number of messages received in the nodes of a clique, the longer it takes for any newly received message to be forwarded to the nodes of the next clique.
The full proof appears in appendix.
\begin{theorem}
	\label{thm:uniform_rand_slow}
	The uniform random algorithm requires $\Omega(n/\sqrt{k})$ rounds on $G_{n,k}$, in expectation.

\end{theorem}

Our main result is an algorithm in which the probabilities for sending messages in each round are not fixed, but rather change dynamically during the execution based on how it evolves. Roughly speaking, the probability of sending a message is set according to the number of times it was received, with the goal of giving higher probabilities for less popular messages. The key intuition behind this approach is that nodes can take responsibility for forwarding messages that they receive few times, while they can assume that messages that have been received many times have already been forwarded throughout the network. This way, we aim to combine qualities of both random and static approaches, obtaining an algorithm that is both fast and robust.

This basic approach alone turns out to be insufficient. It allows each message to be sent fast through multiple paths in the network, but it requires an additional mechanism in order to be robust against failures.
Our next step is to augment our algorithm with some additional rounds of communication that allow the paths to change dynamically as the execution unfolds, essentially bypassing faulty nodes. These \emph{shuffle phases} provide fault-tolerance while retaining the efficiency of the algorithm.
We consider a strong failure model, in which links are reliable but nodes fail independently with probability $q$ per round and never recover, and prove the following result, which holds with high probability\footnote{We use the phrase ``with high probability'' (w.h.p.) to indicate that an event happens with probability at least $1-\frac{1}{n^c}$ for a constant $c\geq1$.}.
\begin{theorem-repeat}{thm:main2}
	\autoref{alg:shuffle} completes full information spreading on $G_{n,k}$ in $O\left(\frac{n}{k}\log^3 n\right)$ rounds, for any node failure probability per round $q$, $0\leq q \leq O\left(\frac{k}{n \log^3 n}\right)$, w.h.p.
\end{theorem-repeat}

While our algorithm is general and does not assume any knowledge of the topology of the network, showing that it is fast and robust for $G_{n,k}$ is important as this graph is basically a $k$-vertex-connected generalization of a simple path. This constitutes a first step towards understanding this key question.
By making minor changes to $G_{n,k}$ we can cover additional graphs with same or similar analysis.
We believe that the same approach works for additional families of $k$-vertex-connected graphs.

\vspace{-10pt}
\subsection{Additional Background and Related Work}
One approach for disseminating information that was introduced in \cite{ahlswede2000network} and has been intensively studied (e.g. \cite{li2003linear, ho2003benefits, deb2006algebraic, mosk2006information}) is \emph{network coding}. Instead of simply relaying the packets they receive, the nodes of a network take several packets and combine them together for transmission. An example is \emph{random linear network coding (RLNC)} presented in \cite{ho2006random}. Among its advantages is improving the network's throughput \cite{ho2003benefits}. A conclusion that can be derived from the analysis shown in \cite{haeupler2011analyzing}, is that RLNC spreads the information in $\Theta(n/k)$ rounds, w.h.p.

However, network coding requires sending large coefficients, which do not fit within the restriction on the packet size that is imposed in the Vertex-Congest model.
An additional disadvantage is derived from the fact that decoding is done by solving a system of linear independent equations of $n$ variables, one variable for each of the original messages.
Thus, the decoding process requires the reception of a \emph{sufficient number} of packets by the node, in order to start reproducing the original information. Unfortunately, in most cases, this sufficient number of packets equals the number of original messages, which means that decoding happens only at the end of the process. This issue has supreme importance in applications of broadcasting videos or presentations. For example, when watching online content, one would prefer displaying the downloaded parts of an image immediately on the screen, rather than waiting with an empty screen until the image is fully downloaded.

An almost-optimal algorithm that requires $O(n\log{n}/k)$ rounds with high probability has been shown in~\cite{CHGK14b}. It is based on a preprocessing stage which constructs vertex-disjoint connected dominating sets (CDSs) which are then used in order to route messages in parallel through all the CDSs. However, this algorithm is non-robust for the following reason.
In the basic algorithm the failure of a single node in a CDS suffices to render the entire structure faulty.
This sensitivity can be easily fixed by combining $O(\polylog(n))$ CDSs together into well-connected components and sending information redundantly over each CDS in the component, incurring a cost of only an $O(\polylog(n))$ factor of slowdown in runtime. Nevertheless, the \emph{construction} itself, of the CDS packings, is highly sensitive to failures. It is an important open problem whether CDS packings can be constructed under faults.

Randomized protocols were designed to overcome similar problems of fault-tolerance in various settings \cite{elsasser2009cover, feige1990randomized}, as they are naturally fault-tolerant. The approach taken in this paper, of changing the probabilities of sending messages according to how the execution evolves such that they are inversely proportional to the number of times a message has been received, bears some resemblance and borrows ideas from~\cite{censorgiakkoupis2012fast}, where a fault-tolerant information spreading algorithm was designed
for gossiping, which is a different model of communication.
Apart from the high-level intuition, the model of communication and the implementation and analysis are completely different.

\subsection{Preliminaries}
We assume a network with $n$ nodes that have unique identifiers of $O(\log n)$~bits.
Each node $u$ holds one message, denoted $m_u$.
An \emph{information spreading} algorithm distributes the messages of each node in the network to all other nodes.

In the \emph{Vertex-Congest} model, each node knows its neighbours but does not know the global graph topology.
The execution proceeds in a sequence of synchronous rounds.
In each round, every node generates a packet and sends it to \emph{all} of its neighbours.
The packet size is bounded by $O(\log{n})$ bits and can encapsulate one message, in addition to some header.

An $n$-node graph is said to be \emph{$k$-vertex-connected} if the graph resulting from deleting any (perhaps empty) set of fewer than $k$ vertices remains connected.
In this paper we assume that $k=\omega(\log^3 n)$.
An equivalent definition \cite{menger1927allgemeinen} is that a graph is $k$-vertex-connected if for every pair of its vertices it is possible to find~$k$ vertex-disjoint paths connecting these vertices.

We consider a strong failure model, in which links are reliable but nodes fail independently with probability $q$ per round and never recover.

\section{A Fast Information Spreading Algorithm}
\label{sec:alg}
In this section, we describe our basic information spreading algorithm. We emphasize that the algorithm does not assume anything about the underlying graph, except for a polynomial bound on its size. In particular, the nodes do not know the vertex-connectivity of the graph, nor any additional information about its topology.
Each node $u$ has a set of received messages, whose content at the beginning of round~$t$ is denoted $R_u(t)$.
We use $cnt_{u,v}(t)$ to denote the number of times a node $u$ has received message $m_v$ by the beginning of round $t$.
Denote by $S_u(t)$ the set of messages sent by node $u$ by the beginning of round $t$.
Define $B_u(t)\equiv R_u(t)-S_u(t)$, the set of messages that are known to node $u$ at the beginning of round $t$, but not yet sent.
We refer to~$B_u(t)$ as a logical variable, whose value changes implicitly according to updates in the actual variables $R_u(t)$ and $S_u(t)$.
For every node $u$, we have that $S_u(0)=\emptyset$, $R_u(0)=\{m_u\}$, $cnt_{u,u}(0)=1$, and for each $v\neq u, cnt_{u,v}(0)=0$.

We present an algorithm, \autoref{alg:1}, that consists of two types of phases: a random phase and ranking phases (see \autoref{fig:phases}).
Let $t_{0}$ be the round number at the beginning of the random phase,
and let $\bar{t}_{0}$ be the round number after the random phase.
Let $t_{p}$ be the round number at the beginning of ranking phase~$p$,
and let $\bar{t}_{p}$ be the round number after ranking phase $p$, starting from~$p=1$.
In this algorithm, it holds that $\bar{t}_{p}=t_{p+1}$ for every $p$, and $t_0=1$.
We will later modify this algorithm in \autoref{sec:alg_ft}, where we argue about properties that hold in $\bar{t}_{p}$ and~$t_{p+1}$, separately.
Denote by $\hat{B}_u(t_p)$ the set of node $u$ at time~$t_p$.
Unlike~$B_u(t)$, $\hat{B}_u(t)$ is an actual variable that does not implicitly change according to $R_u(t)$ and $S_u(t)$.
We assign a value to it at the beginning of every phase, that is, $\hat{B}_u(t_p)=B_u(t_p)$, and make sure that its content only gets smaller during a phase.
The parameters $\alpha$ and $d$ are constants that are fixed later, at the end of~\autoref{sec:time_analysis}.
The algorithm runs as follows,
where in each round every node sends a message and receives messages from all of its neighbors:
\begin{enumerate}[label=(\arabic*)]
	\item Single round (Round 0): This is the first round of the algorithm, where every node~$u$ sends the message $m_u$ it has.
	\item Random phase: This is the first phase of the algorithm, which consists of~$\tau=\alpha\log n$ rounds.
	In each round $t$, every node $u$ picks a message to send from $\hat{B}_u(t_0)$ uniformly at random, and removes it from the set.
	\item Consecutive ranking phases: Each of these phases consists of $\tau^\prime=8 d \tau\log^{2} n$ rounds.
	At the beginning of such a phase, each node uses the Ranking Function (\autoref{alg:ranking}) that defines a probability space over the messages in $\hat{B}_u(t_p)$.
	In each round, every node $u$ picks a message to send from $\hat{B}_u(t_p)$ according to the probability space, and removes it from the set.
\end{enumerate}

\begin{algorithm}[htb!]
	\caption{for each node $u$}
	\label{alg:1}
	\begin{algorithmic}[1]
		\State \Call{SyncRound}{$m_u$} \Comment{Round 0}
		\State RandomPhase()
		\Loop
		\State RankingPhase()
		\EndLoop

		\Algphase{SyncRound($m$)}
		\Procedure{SyncRound}{$m$} \Comment{A synchronized round}
		\State send($m$)
		\State $S_u(t) \gets S_u(t) \cup \{m\}$
		\State $R \gets$ received messages
		\ForAll{$m_v \in R$}
		\State $R_u(t) \gets R_u(t) \cup \{m_v\}$
		\State $cnt_{u,v}(t) \gets cnt_{u,v}(t) +1$
		\EndFor
		\State $t \gets t+1$
		\EndProcedure

		\Algphase{RandomPhase}
		\State $\hat{B}_u(t_0) \gets B_u(t)$ \Comment{$t=t_0$}
		\Loop \textbf{ $\tau$ times}
		\Comment{$\tau = \alpha\log n$}
		\State $m \gets $ pop message from $\hat{B}_u(t_0)$ uniformly at random
		\State \Call{SyncRound}{$m$}
		\EndLoop
		
		\Algphase{RankingPhase $p$}
		\State $\hat{B}_u(t_{p}) \gets B_u(t)$ \Comment{$t=t_{p}$}
		\State $Prob \gets \Call{RankingFunction}{\hat{B}_u(t_{p})}$

		\Loop \textbf{ $\tau^\prime$ times} \Comment{$\tau^\prime=8 d \tau\log^{2} n$}
		\State $m \gets $ pop message from $\hat{B}_u(t_{p})$ according to $Prob$
		\State Nullify $Prob[m]$ (update $Prob$ accordingly)
		\State \Call{SyncRound}{$m$}
		\EndLoop
		
	\end{algorithmic}
\end{algorithm}

\vspace{-10pt}
\subsubsection{Ranking Function.}
The ranking function (in~\autoref{alg:ranking}) is calculated by each node, and defines a probability space over its messages.
Each node $u$ sorts the messages in $\hat{B}_u$ according to their $cnt$ values, smallest to largest, breaking ties arbitrarily.
Denote by $rank_m$ the position of the message $m$ within the sorted list,
and let $b=|\hat{B}_u|$, be the size of the list.
We consider the probability space in which the probability for a message $m$ with $rank_m=r$ to be picked is~$\frac{1}{r H_b}$.
Namely, the probability is inversely proportional to $r$.
The $b$-th harmonic number, $H_b=\sum_{i=1}^b 1/i$, is a normalization factor (over the whole list of messages).
This means that messages in lower positions (lower $rank_m$ values, implying lower $cnt$ values) are more likely to be picked.

\begin{figure}[htb!]
\begin{algorithmic}[1]
	\Function{RankingFunction}{Buffer $\hat{B}_u$}
	\State $mList\gets$ sort $\hat{B}_u$ increasingly according to $cnt$ values
	\State $b \gets$ length($mList$)
    \LINEFORALL{$1\leq r\leq b$}{$Prob[mList[r]] \gets \frac{1}{rH_b}$}
	\State \Return $Prob$
	\EndFunction
\end{algorithmic}
\caption{The Ranking Function}
\label{alg:ranking}
\end{figure}

Other interesting variants of probability distributions over the messages might work as well. For example, the inverse proportion might be raised to some exponent, and be a function of the $cnt$ values instead of the ranking $r$. Our ranking function was selected as it is very simple, and fits perfectly in~\autoref{lemma:ranking}.
In the algorithm, the probability space used by a node $u$ during a phase is calculated at the start of the phase.
In ranking phases, it is defined according to the~Ranking function.
In the random phase, it is the uniform distribution.
Within a phase, the only modifications in the probability space of a node are done due to the \emph{non-repetitive} sending policy\footnote{There is no point in re-sending messages, as all links are reliable.}, i.e., the need for nullifying probabilities of messages that are already sent.
When a message is sent, the modification can be done, for example, by updating the normalization factor, or alternatively by distributing the probability of the sent message between all other messages (say, proportionally to their current probabilities).
Anyhow, this implies that the probability of each message can only get larger during a phase, as long as it is not sent. Namely, the initial probability of a message (at the beginning of a phase) is a lower bound on its probability for the rest of the phase (as long as it is not sent).
Probabilities are not defined for messages that were not known at the start of a phase, and were first received during the phase, thus these messages have no chance of being sent until the next phase starts.
\vspace{-10pt}
\subsubsection{The Phase Separation Property.}
Changes in $cnt$ values during a phase (due to reception of messages) do not affect the probability space of this phase, as it is calculated only at the start of each phase.
This implies that messages that are first received by a node after the start of the random phase or a ranking phase have zero probability for being sent during that phase, and can be sent by the node only starting from the next phase, when the probability space is recalculated.
We call this \emph{the phase separation property}, and it implies the following:
\vspace{-10pt}
\begin{proposition}
	\label{proposition:p_distance}
	At the start of ranking phase $p$, every message has propagated to a distance of at most $p+1$.
\end{proposition}

The following lemma holds for any node and for a general graph. Its proof appears in appendix.
\begin{lemma}
	\label{lemma:ranking}
	Let $m$ be a message with rank $r\leq 8 \tau$ (recall that $\tau=\alpha\log n$), then $m$ is sent during the ranking phase with probability at least $1-n^{-d}$.
\end{lemma}

\section{Time Analysis for $G_{n,k}$}
\label{sec:time_analysis}
Recall that $G_{n,k}$ is the graph that consists of $n/k$ cliques of size $k$ (assume $n/k$ is an integer), with a matching between every two consecutive cliques (see \autoref{fig:g_nk} in appendix).
Clearly, $G_{n,k}$ is $k$-vertex-connected.

\vspace{-10pt}
\subsubsection{Additional Definitions.}
Denote by $\mathcal{C}$ the set of all cliques.
Recall the enumeration of the cliques,
and denote by $C_i$ clique number $i$, $i\in\{1,\ldots,\frac{n}{k}\}$.
Denote by $C(u)$ the clique that contains node $u$.
A \emph{layer} $L$ is a set of $n/k$ nodes from all distinct cliques that form a path starting in $C_1$ and ending in $C_{n/k}$. We denote by $\mathcal{L}$ the set of all $k$ layers.
The layer $L(u)\in \mathcal{L}$ is the layer that contains node $u$.
Notice that within the same clique, different nodes belong to different layers.

We now analyze the time complexity of the algorithm to spread information over $G_{n,k}$.
For simplicity, we analyze the flow of messages from $C_j$ to $C_i$, where~$j\leq i$.
The opposite direction of flow and its analysis are symmetric.

\begin{theorem}
	\label{th:main1}
	\autoref{alg:1} completes full information spreading on $G_{n,k}$ in $O\left(\frac{n}{k}\log^3 n\right)$ rounds, w.h.p.
\end{theorem}

The theorem is directly proved based on~\autoref{lemma:iteration}, as follows.

\begin{lemma}[Iteration]
	\label{lemma:iteration}
	For every $i, 1\leq i\leq\frac{n}{k}$, every node $u\in C_i$, and every node $v$ such that $v\in C_j$ for some $i-p\leq j \leq i$,
	it holds that $m_v\in R_u(\bar{t}_{p})$, w.h.p.
\end{lemma}

\begin{proof}[Proof of \autoref{th:main1}]
	\autoref{lemma:iteration} shows that by the end of ranking phase $p$, w.h.p. each node $u$ knows all messages $m_v$ originating at distance at most $p$. This implies that full information spreading is completed after $n/k$ phases, since $n/k$ is the diameter of the graph, which proves \autoref{th:main1}.
	\qed
\end{proof}

In the rest of the section we prove~\autoref{lemma:iteration}.
The following definition is useful to indicate that a node shares responsibility for disseminating a message.

\begin{definition}[Fresh message]
	A \emph{fresh} message of a node $u$ at time $t$, is a message $m_v\in R_u(t)$ for which $cnt_{u,v}(t)<T$, for threshold $T=\frac{1}{2}\tau$.
\end{definition}

\vspace{-10pt}
\subsubsection{General Idea of the Proof.}
At the end of round 0, every message $m_v$ is disseminated in its own clique $C(v)$.
Then, we show that by the end of the random phase, each message $m_v$ is sent w.h.p. by a sufficiently large number of nodes $u\in C(v)$, to become non-fresh in all nodes of the clique $C(v)$.
Simultaneously, each of the messages $m_v$ becomes known and fresh in a sufficiently large number of nodes in the neighboring clique.

Then we show that ranking phases shift and preserve this situation.
At the beginning of every ranking phase, every fresh message in a node is also fresh in a sufficiently large number of nodes within the same clique.
During the phase, all of the fresh messages are sent w.h.p., implying that each one of the messages
\begin{inlinelistroman}
	\item \label{prop:1} is disseminated in the clique;
	\item \label{prop:2} is not fresh in nodes of the clique anymore; and
	\item \label{prop:3} is fresh in a sufficiently large number of nodes in the neighboring clique.
\end{inlinelistroman}

The combination of properties \ref{prop:2} and \ref{prop:3} is the crux of the proof.
It guarantees that the process progresses iteratively, as it leads to similar conditions again and again at the beginning of every new ranking phase.
This happens because every node can easily distinguish between a new message received from nodes within the clique (becomes non-fresh by the end of the phase), and a new message received from the neighbor in the neighboring clique (stays fresh at the end of the phase, and should be sent during the next phase).
We emphasize that all of this is done implicitly, without knowing the structure of the network.

This iterative behavior of the combined properties guarantees that every message propagates one additional clique per phase, until full information spreading completes after $O(n/k)$ phases.

Let $t'$, for $0\leq t'\leq \tau-1$, be the time from the first round of the random phase, i.e., $t'=t-t_0$.
The following proposition is immediate from the pseudocode:
\begin{proposition}
	\label{proposition:init_buff}
	At the beginning of the random phase, $\hat{B}_u(t_0)$ for every node $u\in C_i$ contains exactly $k-1$ messages $m_v$ originating at $v\in C_i$, and at most two additional messages, one originating at $v\in C_{i-1} \cap L(u)$, and one originating at $v\in C_{i+1} \cap L(u)$.
	Thus, it holds that $|\hat{B}_u(t_0+t')|= k+1-t'$, for $i=2, 3,\cdots,\frac{n}{k}-1$, and $|\hat{B}_u(t_0+t')|= k-t'$, for $i=1, \frac{n}{k}$.
\end{proposition}

Namely, nodes of inner cliques ($C_i, 1<i<n/k$) start the random phase with $|\hat{B}_u(t_0)|=k+1$, while nodes of cliques $C_1$ and $C_{n/k}$ start the random phase with $|\hat{B}_u(t_0)|=k$.

\subsection{Analysis of the Random Phase}
The following lemma analyzes the initial random phase, and shows that every message $m_v$ is non-fresh in all nodes of $C(v)$ at the end of the random phase:
\begin{lemma}
	\label{lemma:known_promotes}
	At the end of the random phase, for every message $m_v$ and for all nodes $u\in C(v)$, $m_v$ is non-fresh for $u$, with probability at least $1-\frac{1}{n^{\alpha/48-1}}$.
\end{lemma}

\begin{proof}
	Fix $v$.
	Message $m_v$ is disseminated in $C(v)$ by the start of the random phase.
	By \autoref{proposition:init_buff}, for every $u \in C(v)$,
	it holds that $|\hat{B}_u(t_0+t')|\leq k+1-t'$ during the random phase.
	
	Let $\mathbbm{1}_{u,v}$, for every $u\in C(v)$, be an indicator variable that indicates whether node $u$ sends $m_v$ during the random phase or not.
	Then
	\begin{equation*}
		\begin{split}
			\Pr[\mathbbm{1}_{u,v}=1] & \geq 1-\prod_{t'=0}^{\tau-1} \frac{k-t'}{k+1-t'} = 1-\frac{k+1- \tau}{k+1} \geq \frac{\tau}{(3/2)k} \enspace .
		\end{split}
	\end{equation*}
	Let $X_{v}=\sum_{u\in C(v)} \mathbbm{1}_{u,v}$, be the number of nodes in $C(v)$ that send $m_v$ during the random phase, i.e., the number of times $m_v$ is received by every node in $C(v)$.
	Then
	\begin{equation*}
		\begin{split}
			\mu =E(X_{v}) &=E\left(\sum_{u\in C(v)} \mathbbm{1}_{u,v}\right)=\sum_{u\in C(v)} E(\mathbbm{1}_{u,v}) \geq\sum_{u\in C(v)} \frac{2\tau}{3k} = \frac{2\tau}{3} \enspace .
		\end{split}
	\end{equation*}
	Since $v$ is fixed, the indicator variables are independent, as they refer to decisions of distinct nodes.
	By applying a Chernoff bound~\cite[Chapter~4]{mitzenmacher2005probability}, we get
	$$\Pr[X_{v}\leq(1-\delta)\mu] \leq \exp\left(-\delta^2\mu/2\right) \leq \exp\left(-\delta^2\alpha\log n/3\right) < 1/n^{\frac{\alpha\delta^2}{3}} \enspace .$$
	By setting $\delta=\frac{1}{4}$, we get that a message $m_v$ is non-fresh in all nodes $u\in C(v)$ with probability at least $1- \frac{1}{n^{\alpha/48}}$.
	By a union bound, this holds for every node $v$ with probability at least $1-\frac{1}{n^{\alpha/48-1}}$.
	\qed
\end{proof}

\begin{definition}
	A \emph{pioneer} message in node $u\in C_{i}$ at time $t_p$ (beginning of ranking phase $p$), is a message $m_v\in R_u(t_p)$ that originated at $v\in C_{i-p-1}$.
\end{definition}

\paragraph{Pioneer Attributes.} If a message $m_v$ is a pioneer in node $u\in C_i$ at time $t_p$, then
\begin{inlinelistroman}
	\item \label{pio:layer} $v\in L(u)$ (by \autoref{proposition:p_distance}, the message was transmitted over the shortest path), and the following hold at time $t_p$:
	\item \label{pio:cnt1} $cnt_{u,v}(t_p) = 1$, and thus $m_v$ is fresh for $u$,
	\item \label{pio:unique} $m_v\notin R_{u'}(t_p)$ for every $u' \in C_{i}, u'\neq u$ (by \autoref{proposition:p_distance}),
	\item \label{pio:known} $m_v$ is disseminated in $C_{i-1}$ (by the node that relayed $m_v$ to its neighbor in $C_i$),
	and \item \label{pio:fresh} $m_v$ is fresh in every node $u'\in C_{i-1}$.
\end{inlinelistroman}
The following is proved in appendix.

\begin{lemma}
	\label{lemma:rand_pioneer}
	With probability at least $1-1/n^{\alpha/24-1}$,
at the end of the random phase, for every $i$, the number of pioneer messages that reach $C_i$ is $\leq 3 \tau$.
\end{lemma}

\subsection{Analysis of Ranking Phases}
After analyzing the single random phase, here we analyze the ranking phases.
\begin{lemma}
	\label{lemma:h_send}
		With probability at least $1-\frac{1}{n^{d-2}}$,
	every node $u$ that starts ranking phase $p$ with at most $8 \tau$ fresh messages, sends all of them during the phase.
\end{lemma}
The proof appears in appendix.
To prove \autoref{lemma:iteration}, we show a sequence of four inductive properties, that hold for ranking phase $p$, with probability at least $1-\left(\frac{2p}{n^{d-2}}+\frac{2}{n^{\alpha/48-1}}\right)$.

\vspace{-10pt}
\subsubsection{Property 1.}
For every $i, 1\leq i\leq\frac{n}{k}$, it holds that the number of messages $m_v$, $v\in C_{i-p-1}$, such that $m_v \in R_u(t_{p})$ for some $u\in C_i$ (pioneers), is at most $3 \tau$, and each reaches a distinct node $u\in L(v)$.

\vspace{-10pt}
\subsubsection{Property 2.}
For every $i, 1\leq i\leq\frac{n}{k}$, and every node $u\in C_i$, it holds that at time $t_{p}$ there are at most $4 \tau$ \emph{fresh} messages $m_v$ for node $u$ for every one of the two directions of flow ($8 \tau$ in total).
All of them originated at nodes $v\in C_{i-p}$ (similarly, $v\in C_{i+p}$), except for at most one (a pioneer) which originated at $u'\in C_{i-p-1} \cap L(u)$ (similarly, $u'\in C_{i+p+1} \cap L(u)$).
All messages $m_v\in R_u(t_{p}), v\in C_{i-p}$ (similarly, $v\in C_{i+p}$), are fresh.

\vspace{-10pt}
\subsubsection{Property 3.}
For every $i, 1\leq i\leq\frac{n}{k}$, and every node $v\in C_{i-p}$, it holds that $m_v$ is fresh for at least $T$ nodes $u\in C_i$ at time $t_{p}$.
Recall that $T=\tau/2$.

\vspace{-10pt}
\subsubsection{Property 4.}
For every $i, 1\leq i\leq\frac{n}{k}$, every node $u\in C_i$, and every node $v$ such that $v\in C_j$ for some $i-p\leq j \leq i$,
it holds that $m_v\in R_u(\bar{t}_p)$, and $m_v$ is non-fresh.

\paragraph{}
\vspace{-5pt}
We prove the four properties simultaneously by induction on the ranking phase number, $p$.
To prove the base cases, we assume that all events described in \autoref{lemma:known_promotes}, \autoref{lemma:rand_pioneer}, and \autoref{lemma:h_send} (for $p=1$) occur. Notice that, by a union bound, the probability for this is at least $1-\left(\frac{1}{n^{\alpha/24-1}}+\frac{1}{n^{\alpha/48-1}}+\frac{1}{n^{d-2}}\right)\geq 1-\left(\frac{2}{n^{\alpha/48-1}}+\frac{2}{n^{d-2}}\right)$.

To prove the induction step, we assume that all events described in the four properties for $p-1$, and in \autoref{lemma:h_send} for $p-1$ and $p$, occur.
This happens with probability at least $1-\left(\frac{2}{n^{\alpha/48-1}}+\frac{2(p-1)}{n^{d-2}}+\frac{1}{n^{d-2}}+\frac{1}{n^{d-2}}\right)=1-\left(\frac{2}{n^{\alpha/48-1}}+\frac{2p}{n^{d-2}}\right)$.
The complete inductive proof appears in appendix.
Property 4 guarantees that full information spreading is completed after ranking phase $p=n/k$, with probability at least $1-\left(\frac{2n/k}{n^{d-2}}+\frac{2}{n^{\alpha/48-1}}\right)\geq 1-\left(\frac{1}{n^{d-3}}+\frac{1}{n^{\alpha/48-2}}\right) \geq 1-\frac{1}{n^c}$, for a constant $c$, by fixing $d$ and $\alpha$ to values $d>c+3, \alpha >48c+96$.
This completes the proof of \autoref{lemma:iteration}, from which \autoref{th:main1} follows.

\section{Fault Tolerance}
\label{sec:alg_ft}
\autoref{alg:1} highly depends on the random phase in the following sense.
For every node $v$, consider the set of nodes in neighboring cliques that know message $m_v$ by the end of the random phase.
Then, w.h.p. the algorithm spreads $m_v$ using the layers of nodes in the above set (``carriers'').
This means that the paths of a message are fixed very early in the algorithm and do not alternate.

A single failure of a node in each layer (carrier) is sufficient to break down its role.
Each message relies on at least $T$ different layers to proceed.
Hence, the algorithm is sensitive to failures in which less than $T$ carrier layers are non-faulty.

At the beginning of ranking phase $p$, consider the case where a message $m_v\in C_{i-p}$ is fresh in $x<T$ nodes in clique $C_i$, due to failures.
The behavior of the algorithm in such case is as follows:
During the ranking phase, less than $T$ nodes in the clique send the message, so all other nodes in $C_i$ receive the message less than $T$ times, thus it stays fresh in all of them at the end of ranking phase $p$.
Starting from the next ranking phase, the message $m_v$ propagates regularly over those $x<T$ carriers, but also propagates over all other carriers, with a delay of a phase.
This means that every layer becomes responsible for one extra message (in addition to at most $8 \tau$ messages),
which may still be tolerable.
In general, our algorithm can manage a constant number of such occurrences.

We aim to cope with a larger number of failures, so we modify our algorithm to help layers bypass their failing nodes, so they continue operating as carriers.

\vspace{-3mm}
\subsection{Shuffle Phases}
We invoke a shuffle phase between every two ranking phases, so phases of the algorithm now proceed as described in \autoref{fig:phases_2}.
\begin{algorithm}[htb!]
	\caption{for each node $u$}
	\label{alg:shuffle}
	\begin{algorithmic}[1]
		\State \Call{SyncRound}{$m_u$} \Comment{Round 0}
		\State RandomPhase()
		\Loop
		\State RankingPhase()
		\State ShufflePhase()
		\EndLoop
		\Algphase{ShufflePhase $p$}
		\State $\hat{B}_u(\bar{t}_{p}) \gets$ fresh messages in $B_u(t)$ \Comment{$t=\bar{t}_{p}$}
		\ForAll{$m_v \in \hat{B}_u(\bar{t}_{p})$}
		\State $phasecnt_{u,v} \gets 1$
		\EndFor
		\State $R \gets \hat{B}_u(\bar{t}_{p})$
		\Loop \textbf{ $8\tau$ times}
		\If{$\hat{B}_u(\bar{t}_{p}) = \emptyset$}
		\State send own message $m_u$
		\Else
		\State pop and send a fresh message from $\hat{B}_u(\bar{t}_{p})$
		\EndIf
		\State $R' \gets$ receive messages
		\ForAll{$m_v \in R'$}
		\If {$m_v \notin R$}
		\State $phasecnt_{u,v} \gets 1$
		\Else
		\State $phasecnt_{u,v} \gets phasecnt_{u,v} +1$
		\EndIf
		\State $R \gets R \cup \{m_v\}$
		\EndFor
		\State $t \gets t+1$
		\EndLoop
		\State $R \gets$ $R$ after filtering out unwanted messages. \label{alg2:lineR}
		\Comment{Filter out messages $m_v$ with $phasecnt_{u,v}<\hat{c}\cdot T$}
		\Comment{Filter out messages that were non-fresh prior to the start of the phase}
		\State $R_u(t) \gets R_u(t) \cup R$
		\State Select $4 \tau$ messages from $R$ randomly, rank them from 1 to $4 \tau$.
	\end{algorithmic}
\end{algorithm}
Roughly speaking, the objective of a shuffle phase, is that nodes of every clique re-divide their responsibilities over messages.

A shuffle phase consists of $8 \tau$ rounds.
During it, every node sends its fresh messages (and receives fresh messages from all neighbors).
Instead of updating the regular $cnt$ values, nodes use separate counters,~$phasecnt$, to count the number of receptions for each message during the current shuffle phase.
Recall that the objective is shuffling the fresh messages between nodes of same clique.
Thus, at the end the of the shuffle phase, every node identifies and filters out unwanted messages, which are messages received from neighboring cliques (low $phasecnt$ values), and messages that were already non-fresh prior to the start of the shuffle phase.
Then it randomly picks $4 \tau$ new fresh messages, to start the next ranking phase with.

The important gain from this cooperative division of responsibilities done by the nodes of a clique, is that a node $u\in C_i$ that does not receive new messages from its faulty neighbor $u'\in C_{i-1} \cap L(u)$, can overcome the failure of the carrier layer, and still take part in transmitting relevant messages from one clique to the other, with no delays.
The proof of the following appears in appendix.
\begin{theorem}
	\label{thm:alg2_time}
		\autoref{alg:shuffle} completes full information spreading on $G_{n,k}$ in $O\left(\frac{n}{k}\log^3 n\right)$ rounds, w.h.p.
\end{theorem}

\vspace{-3mm}
\subsection{Resilience to Faults}
Recall that we consider a model of independent failures of nodes, where each node fails at each round with probability $q$, and never recovers.
Let ${\tau_{e}}\leq2\frac{n}{k}\tau^\prime=O\left(\frac{n}{k}\log^3 n\right)$ (the round number at the end of ranking phase $n/k$ in \autoref{alg:shuffle}).
First, we prove the following.
The proof appears in appendix.
\begin{lemma}
	\label{lemma:nonfaulty_clique}
		At the end of round ${\tau_{e}}$, the number of non-faulty nodes in each clique is at least $(30k/32)$, with probability at least $1-1/n^{30}$.
\end{lemma}

We show that the algorithm tolerates failures for $q$, $0\leq q \leq O\left(\frac{k}{n \log^3 n}\right)$.
\begin{theorem}
	\label{thm:main2}
	\autoref{alg:shuffle} completes full information spreading on $G_{n,k}$ in $O\left(\frac{n}{k}\log^3 n\right)$ rounds, for any node failure probability per round $q$, $0\leq q \leq O\left(\frac{k}{n \log^3 n}\right)$, w.h.p.
\end{theorem}
\begin{proof}
Fix $i, p$.
Let $m_v$ be a message that is fresh in at least $T$ (non-faulty) nodes in $C_{i-1}$ at the end of shuffle phase $p-1$.
Here we analyze the probability that $m_v$ is \emph{not} shuffled successfully in clique $C_i$.

An unsuccessful shuffle might occur either because the $phasecnt$ values in $C_i$ at the end of shuffle phase $p$ are smaller than the threshold of $T^{*}=\hat{c}T$, so the message is filtered out (denote this event by $A$), or because the message was selected by less than $T$ (non-faulty) nodes.
By~\autoref{lemma:h_send}, at the beginning of shuffle phase $p$, the message $m_v$ is supposed to be fresh in at least $T$ nodes in $C_i$ (each of them gets the message from its respective neighbor in $C_{i-1}$).
Of these nodes in $C_i$, if one does not send $m_v$ during shuffle phase $p$, then either the node or its neighbor in $C_{i-1}$ (or both) becomes faulty by the end of shuffle phase~$p$.
The probability $\hat{q}$ for such a pair of nodes \emph{not} to fail is bounded from below (according to Bernoulli's inequality) by $\hat{q} = ((1-q)^{{\tau_{e}}})^2 \geq (1-q{\tau_{e}})^2\geq 1-2q{\tau_{e}}\geq1-1/16$.

Fix a set of $T$ pairs of nodes $S(m_v)\subseteq C_{i-1} \times C_i$, of those who know message $m_v$ in $C_{i-1}$ at the end of shuffle phase $p-1$, and their respective neighbors in $C_i$.
There might exist more than $T$ such pairs, but by fixing a set of size $T$ and ignoring the rest, we bound the probability of an unsuccessful shuffle from above, as the ignored nodes can only help and increase the probability of success.
A ``surviving'' pair is a pair of nodes from $S(m_v)$ where both are non-faulty at the end of the shuffle phase, and hence function properly (by sending message $m_v$) during shuffle phase $p$.
Denote by $s$, the number of ``surviving'' pairs.
We have:
\begin{equation*}
\begin{split}
		\Pr[A] & \leq \sum\limits_{s=0}^{T^{*}-1} {{T \choose s} \hat{q}^s (1-\hat{q})^{T-s}} \leq \sum\limits_{s=0}^{T^{*}-1} {{T \choose s} (1-\hat{q})^{T-s}} \leq \sum\limits_{s=0}^{T^{*}-1} {{T \choose s}  \left(\frac{1}{16}\right)^{T-s}} \enspace .
\end{split}
\end{equation*}
We sum over all $s\in\{0 , \ldots, T^{*}-1 \}$, where the number of ``survivors'' is lower than the threshold of~$\hat{c}T$, which implies that the message $m_v$ is filtered out, improperly, at the end of the shuffle phase due to a low $phasecnt$ value.

By setting $0 < \hat{c}\leq\frac{1}{2}$, we get that $\Pr[A] \leq 1/n^{\alpha/3-1}$ (see calculation in appendix).
Namely, the message is not filtered out with probability at least $1/n^{\alpha/3-1}$.
The number of non-faulty nodes in each clique is at least $31k/32$ with probability at least $1-\frac{1}{n^{30}}$, by~\autoref{lemma:nonfaulty_clique}.
An analysis similar to the one in the proof of~\autoref{lemma:shuffle_threshold} (with $\delta=11/15$) gives that,
once the message is not filtered out, it is selected by at least $T$ of the non-faulty nodes in $C_i$ with probability at least $1-1/n^{11^2\alpha/(15\cdot16)}$.
In total, by using a union bound, a message is not shuffled successfully between two consecutive shuffle phases with probability at most $\frac{1}{n^{\alpha/3-1}}+\frac{1}{n^{11^2\alpha/(15\cdot16)}}+\frac{1}{n^{30}} \leq \frac{1}{n^{6}}$ (for value of $\alpha$ fixed earlier).

We use union bound two more times, for all messages and for all phases, and get an upper bound for the probability that a message is not propagated properly, of $\frac{1}{n^{4}}$.
This proves that the algorithm tolerates failures that occur with probability $0\leq q\leq \frac{1}{32{\tau_{e}}}$ in the given model, with probability at least $1-\frac{1}{n^{4}}$.
\qed
\end{proof}

\section{Discussion}
\label{sec:discussion}
\subsubsection{Static-Routes Algorithms.}
Let $ALG$ be an algorithm that spreads information on $k$-vertex-connected graphs in $O\left(\frac{n}{k}\cdot\polylog(n)\right)$ rounds, by constructing static routes, and using them to disseminate messages in parallel, each message on a specific route.
This makes $ALG$ very sensitive to failures, as a single failure in a route suffices to render the entire route faulty.

However, it can easily be configured so that vertex-disjoint routes are combined into groups of size $\gamma$, and every node duplicates its messages and sends them concurrently over these components.
Notice that in $k$-vertex-connected graphs, $\gamma$ is bounded from above by $k$.
This costs $\gamma$ slowdown in runtime as a trade-off.
Denote this configuration of the algorithm by $ALG(\gamma)$.

We are interested in cases where $\gamma=O(\polylog(n))$, so that the runtime of the algorithm remains $O\left(\frac{n}{k}\cdot \polylog(n)\right)$.
Every combination of $\gamma$ vertex-disjoint routes induces a $\gamma$-vertex-connected subgraph, as it stays connected after the removal of any $\gamma-1$ vertices.
Each component functions as long as it stays connected.
According to \cite[Theorem~1.5]{CHGK14a}, for $\gamma=\Omega(\log^3n)$,
such a component stays connected w.h.p.
if its nodes are \emph{sampled} independently with a constant probability.
By considering the sampling process imposed by failures, i.e. considering the non-faulty nodes as sampled, then each component stays connected if a constant fraction of its nodes stays non-faulty during the execution, tolerating a constant fraction of nodes that fail.
The additional slowdown factor for each message to spread over such a component in the presence of faults can be loosely bounded form above by $O(\gamma)$, as the size of the combined component is $O(\gamma)$ the size of its original routes, (in the worst case a message traverses over all non-faulty nodes of the component).
In total, this configuration of the algorithm tolerates the failure of a constant fraction of nodes during its execution, which matches a probability of failure of $q=O\left(\frac{k}{n\cdot \polylog(n)}\right)$ per round, while preserving a time complexity of $O\left(\frac{n}{k}\cdot \polylog(n)\right)$.

The algorithm presented in~\cite{CHGK14b} is static-route, as it constructs CDS packings and routes messages over them.
The CDS packings are only fractionally vertex-disjoint, which requires a few modifications to the above analysis.
However, despite the above fix, the algorithm remains vulnerable due to the preprocessing stage.
Tolerating failures that occur during the preprocessing stage is more complicated, and the construction of CDS packings in the presence of failures is still an open problem.

\vspace{-10pt}
\subsubsection{Summary.}
In this paper, we show an information spreading algorithm, and prove that it is fast and robust for $G_{n,k}$.
The intriguing open question is whether this approach can work for general $k$-vertex-connected graphs.

To summarize, we find the question of devising a fast and robust information spreading algorithm in the Vertex-Congest model an intriguing open question, and view our result as a first step in this direction. The technique our algorithm leverages, of using probability distributions that change over time according to how the execution unfolds, may have applications in other settings as well.

\paragraph{Acknowledgements:} Keren Censor-Hillel is a Shalon Fellow. This research is supported by the Israel Science Foundation (grant number 1696/14).
We thank Mohsen Ghaffari, Fabian Kuhn, Yuval Emek and Shmuel Zaks for useful discussions.

\bibliographystyle{llncs2e/splncs}
\bibliography{back/ref}

\begin{thebibliography}{10}
\providecommand{\url}[1]{\texttt{#1}}
\providecommand{\urlprefix}{URL }

\bibitem{ahlswede2000network}
Ahlswede, R., Cai, N., Li, S.Y., Yeung, R.W.: Network information flow. IEEE
  Transactions on Information Theory  46(4),  1204--1216 (2000)

\bibitem{CHGK14b}
Censor-Hillel, K., Ghaffari, M., Kuhn, F.: Distributed connectivity
  decomposition. In: Proceedings of the 33rd ACM Symposium on Principles of
  Distributed Computing. pp. 156--165. PODC (2014)

\bibitem{CHGK14a}
Censor-Hillel, K., Ghaffari, M., Kuhn, F.: A new perspective on vertex
  connectivity. In: Proceedings of the Twenty-Fifth Annual ACM-SIAM Symposium
  on Discrete Algorithms. pp. 546--561. SODA (2014),
  \url{http://epubs.siam.org/doi/abs/10.1137/1.9781611973402.41}

\bibitem{censorgiakkoupis2012fast}
Censor-Hillel, K., Giakkoupis, G.: Fast and robust information spreading.
  Unpublished manuscript  (2012)

\bibitem{deb2006algebraic}
Deb, S., M{\'e}dard, M., Choute, C.: Algebraic gossip: A network coding
  approach to optimal multiple rumor mongering. IEEE Transactions on
  Information Theory  52(6),  2486--2507 (2006)

\bibitem{elsasser2009cover}
Els{\"a}sser, R., Sauerwald, T.: Cover time and broadcast time. In: Proceedings
  of the 26th International Symposium on Theoretical Aspects of Computer
  Science, STACS. pp. 373--384 (2009)

\bibitem{feige1990randomized}
Feige, U., Peleg, D., Raghavan, P., Upfal, E.: Randomized broadcast in
  networks. Random Structures \& Algorithms  1(4),  447--460 (1990)

\bibitem{haeupler2011analyzing}
Haeupler, B.: Analyzing network coding gossip made easy. In: Proceedings of the
  43rd annual ACM symposium on Theory of computing. pp. 293--302. STOC (2011)

\bibitem{ho2003benefits}
Ho, T., Koetter, R., Medard, M., Karger, D.R., Effros, M.: The benefits of
  coding over routing in a randomized setting. In: Proceedings of the IEEE
  International Symposium on Information Theory. p. 442 (2003)

\bibitem{ho2006random}
Ho, T., M{\'e}dard, M., Koetter, R., Karger, D.R., Effros, M., Shi, J., Leong,
  B.: A random linear network coding approach to multicast. IEEE Transactions
  on Information Theory  52(10),  4413--4430 (2006)

\bibitem{KLN11}
Kuhn, F., Lynch, N., Newport, C.: The abstract {MAC} layer. Distributed
  Computing  24(3-4),  187--206 (2011),
  \url{http://dx.doi.org/10.1007/s00446-010-0118-0}

\bibitem{li2003linear}
Li, S.Y., Yeung, R.W., Cai, N.: Linear network coding. IEEE Transactions on
  Information Theory  49(2),  371--381 (2003)

\bibitem{menger1927allgemeinen}
Menger, K.: Zur allgemeinen kurventheorie. Fundamenta Mathematicae  10(1),
  96--115 (1927)

\bibitem{mitzenmacher2005probability}
Mitzenmacher, M., Upfal, E.: Probability and computing: Randomized algorithms
  and probabilistic analysis. Cambridge University Press (2005)

\bibitem{mosk2006information}
Mosk-Aoyama, D., Shah, D.: Information dissemination via network coding. In:
  2006 IEEE International Symposium on Information Theory. pp. 1748--1752. IEEE
  (2006)

\bibitem{peleg00}
Peleg, D.: Distributed Computing: A Locality-Sensitive Approach. SIAM (2000)

\end{thebibliography}

\section{Appendix}

\subsection{The Graph $G_{n,k}$}
\label{sec:Gnk}

\begin{figure}
	\centering
	\includegraphics[width=\textwidth]{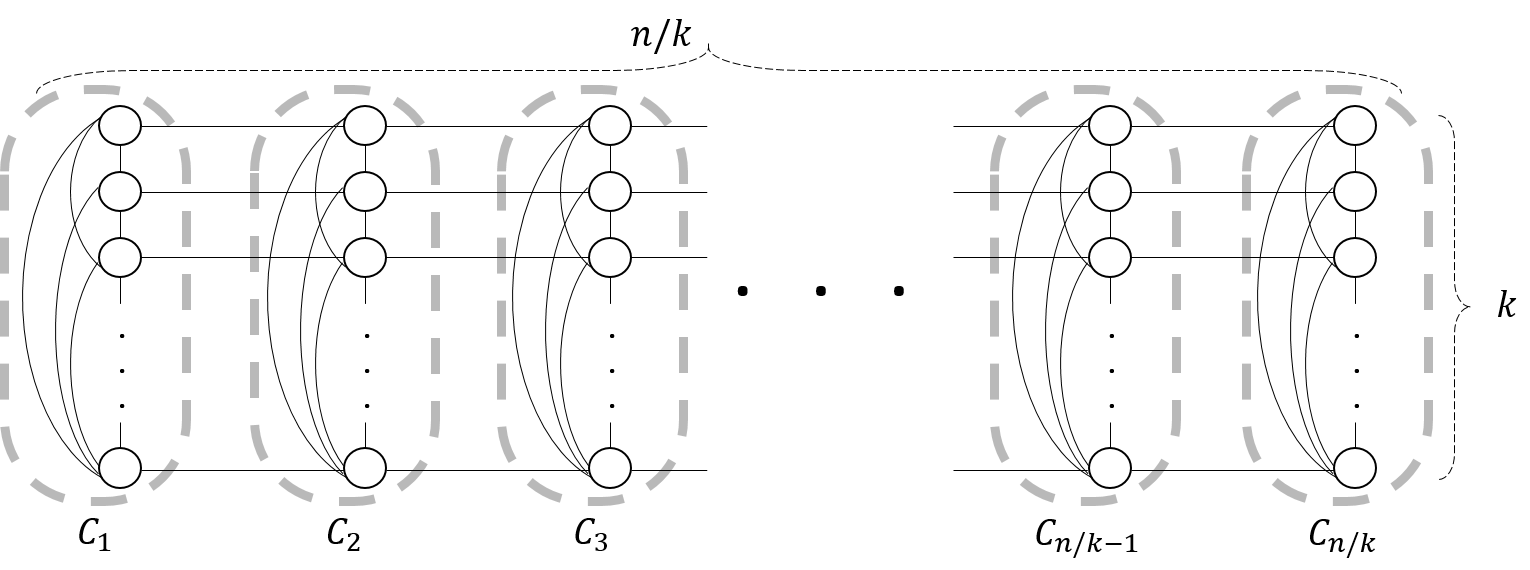}
	\caption{$G_{n,k}$ is an example of a $k$-vertex-connected graph with diameter $\frac{n}{k}$.}
	\label{fig:g_nk}
\end{figure}

\subsection{The Uniform Random Algorithm}
\label{sec:rand}
We consider the uniform random algorithm, in which every node picks and sends a message from its buffer in each round uniformly at random.
We show that the time complexity of the algorithm is asymptotically much slower than the optimal $\Omega(n/k)$.
Consider the uniform random algorithm running on graph $G_{n,k}$. We prove that the expected number of rounds for full information spreading is $\Omega(n/\sqrt{k})$.
First, we prove the following.
\begin{lemma}
	\label{lemma:linear_buff}
	If the buffer size of a node is at least $n/4$, then the number of rounds needed for a message $m_v$ in its buffer to be sent is $\Theta(n)$ in expectation.
\end{lemma}
\begin{proof}
	During the first $n/8$ rounds, the buffer size is at least $n/4-n/8=n/8$.
	The number of rounds until the message $m_v$ is first sent is bounded from below by a geometric random variable with success probability $p= 8/n$.
	The expectation of a geometric random variable is $1/p=n/8$, and the lemma follows.
	\qed
\end{proof}

\begin{theorem-repeat}{thm:uniform_rand_slow}
	The uniform random algorithm requires $\Omega(n/\sqrt{k})$ rounds on $G_{n,k}$, in expectation.

\end{theorem-repeat}
\begin{proof}
	To prove the theorem, we define a partition over the whole space, and calculate the conditional expectations of the number of rounds for each case.
	We show that the expected number of rounds in every case is $\Omega(n/\sqrt{k})$, and the theorem follows according to the law of total expectation:
	$$ E[X] = \sum_{i} E[X \mid A_i] \Pr[A_i] \enspace ,$$
	where $\{A_i\}$ is the partition.
	
	Let $r_0$ be the random variable of the first round number in which the buffer size of all nodes in clique $C_{n/k}$ is at least $n/2$.
	The buffer of every node consists of the messages it knows but not sent yet. In the executions in which such round does not exist, it holds that $n-t \leq n/2$, where $t$ is the last round of the dissemination process, implying $t\geq n/2\geq n/\sqrt{k}$.
	
	Otherwise, $r_0$ is well defined, and it holds that
	\begin{equation}
	E\left[r_0 \ \middle|\ r_0\geq\frac{n}{\sqrt{k}}\right] \geq \frac{n}{32\sqrt{k}} \enspace .
	\end{equation}
	Consider the set of messages $M_1=\{m_v \mid m_v$ is known to some node $u\in C_{n/k}\}$.
	We analyze two possible cases:
	
	\begin{enumerate}
		\item If $|M_1|<n$ at $r_0$, then there exists a message that is not known to any node in $C_{n/k}$ at round $r_0$. Let $m_v$ be such a message.
		Let $r_1$ be the random variable of the number of rounds since $r_0$ until the message $m_v$ spreads to all nodes of $C_{n/k}$.
		We argue that $E[r_0+r_1]\geq \frac{n}{32\sqrt{k}}$.
		The following trivially hold, since $r_0$ and $r_1$ are non-negative:
		\begin{equation}
		E\left[r_0+r_1 \ \middle|\ r_0<\frac{n}{\sqrt{k}}, r_1\geq \frac{n}{24}\right] \geq \frac{n}{32\sqrt{k}} \enspace .
		\end{equation}
		To conclude the argument for this case, it is enough to show that
		\begin{equation}
		\label{eq:3}
		E\left[r_0+r_1 \ \middle|\ r_0<\frac{n}{\sqrt{k}}, r_1<\frac{n}{24}\right]\geq \frac{n}{32\sqrt{k}} \enspace .
		\end{equation}
		
		In the following, we assume that $r_0<n/\sqrt{k}$ and $r_1<n/24$, and give a lower bound for $E\left[r_1 \ \middle|\ r_0<\frac{n}{\sqrt{k}}, r_1<\frac{n}{24}\right]$.
		In addition, we assume that $k\leq n/6$.
		At round $r_0$ it holds that at least $n/2-k\geq n/3$ messages are disseminated in $C_{n/k-1}$, for otherwise the messages do not reach nodes of $C_{n/k}$.
		During the $r_1$ rounds in the interval $[r_0, r_0+r_1]$, all buffers in all nodes in $C_{n/k}$ are of size at least $n/2-r_1$, in all nodes in $C_{n/k-1}$ are of size at least $n/3-r_1$, which means that all buffer sizes of nodes both cliques	are at least $n/4$ during the~$r_1$ rounds in the interval.
		Since buffer sizes are at least $n/4$ and at most $n$, the probability $\hat{q}$ that a node $v\in C_{n/k-1} \cup C_{n/k}$ that knows $m_v$ sends it is $1/n\leq \hat{q}\leq 4/n$.
		Let $X_r$ be the number of nodes in $C_{n/k-1}$ that send $m_v$ during the $r$ rounds that follow round $r_0$.
		It holds that $E\left[X_r \ \middle|\ r_0<\frac{n}{\sqrt{k}}, r_1<\frac{n}{24}\right]$ is at most $k(1-(1-\hat{q})^{r})\leq k(1-(1-\hat{q}{r})) =k\hat{q}{r}\leq \frac{4k{r}}{n}$ (the first inequality is according to Bernoulli's inequality).
		Each node in $C_{n/k-1}$ that sends the message relays it to its corresponding neighbor in $C_{n/k}$.
		If any of these nodes in $C_{n/k}$ sends $m_v$, then the message is disseminated and all $k$ nodes in $C_{n/k}$ know it.
		
		For every $r\geq n/k$, denote by $A_r$ the event $X_r \leq \frac{8kr}{n}$. By applying Markov's inequality we get that $\Pr\left[X_r \geq \frac{8kr}{n} \ \middle|\ r_0<\frac{n}{\sqrt{k}}, r_1<\frac{n}{24}\right]\leq \frac{4kr}{n}\big/ \frac{8kr}{n}=\frac{1}{2}$, and hence,
		$$\Pr\left[A_r \ \middle|\  r_0<\frac{n}{\sqrt{k}}, r_1<\frac{n}{24}\right] \geq \frac{1}{2} \enspace .$$
		Under the assumption that $A_r$ occurs, the probability for the dissemination to occur during these $r$ rounds (implying that $r_1\leq r$)
		is at most $1-\left( (1-\hat{q})^r\right)^{8kr/n} = 1-(1-\hat{q})^{8kr^2/n} \leq 1-(1-8kr^2\hat{q}/n)=8kr^2\hat{q}/n\leq \frac{32kr^2}{n^2}$, (first inequality is according to Bernoulli's inequality).
		By assigning $r=\frac{n}{8\sqrt{k}}>n/k$, we get that
		$$\Pr\left[r_1 > \frac{n}{8\sqrt{k}} \ \middle|\ A_r, r_0<\frac{n}{\sqrt{k}}, r_1<\frac{n}{24}\right] \geq 1-\frac{32k{(n/8\sqrt{k})}^2}{n^2}\geq \frac{1}{2} \enspace ,$$
		and hence
		\begin{align*}
		\Pr&\left[r_1 > \frac{n}{8\sqrt{k}} \ \middle|\  r_0<\frac{n}{\sqrt{k}}, r_1<\frac{n}{24}\right]   \geq \\
		& \geq \Pr\left[r_1 > \frac{n}{8\sqrt{k}}, A_r \ \middle|\  r_0<\frac{n}{\sqrt{k}}, r_1<\frac{n}{24}\right] = \\
		&= \Pr\left[r_1 > \frac{n}{8\sqrt{k}} \ \middle|\ A_r, r_0<\frac{n}{\sqrt{k}}, r_1<\frac{n}{24}\right]\cdot \Pr\left[A_r \ \middle|\  r_0<\frac{n}{\sqrt{k}}, r_1<\frac{n}{24}\right] \geq \\
		&\geq \frac{1}{2}\cdot\frac{1}{2} = \frac{1}{4} \enspace ,
		\end{align*}
		the first equality is according to the law of conditional probability, $P(A \cap B) = P(A|B)P(B)$.
		This gives
		$$E\left[r_1 \ \middle|\ r_0<\frac{n}{\sqrt{k}}, r_1<\frac{n}{24}\right] \geq \frac{n}{32\sqrt{k}} \enspace ,$$
		which proves \eqref{eq:3}.
		
		\item If $|M_1|=n$ at $r_0$,
		let $r_2$ be the random variable of the number of rounds since $r_0$ until all messages are disseminated in $C_{n/k}$.
		We argue that $E[r_0+r_2]\geq \frac{n}{32\sqrt{k}}$.
		Since $r_0$ and $r_2$ are non-negative, the following hold trivially:
		\begin{equation}
		E\left[r_0+r_2 \ \middle|\ r_0 < \frac{n}{\sqrt{k}} , r_2\geq\frac{n}{24}\right] \geq \frac{n}{32\sqrt{k}} \enspace .
		\end{equation}
		To conclude the argument for this case, it is enough to show that
		\begin{equation}
		\label{eq:6}
		E\left[r_0+r_2 \ \middle|\ r_0<\frac{n}{\sqrt{k}}, r_2<\frac{n}{24}\right]\geq \frac{n}{32\sqrt{k}} \enspace .
		\end{equation}
		
		In the following, we assume that $r_0<n/\sqrt{k}$ and $r_2<n/24$, and give a lower bound for $E\left[r_2 \ \middle|\ r_0<\frac{n}{\sqrt{k}}, r_2<\frac{n}{24}\right]$.
		In addition, we assume that $k\leq n/6$.
		In each round, a node in $C_{n/k}$ receives at most $k$ new messages and sends one, and hence the buffer size can increase by at most $k-1$ in a single round. By its definition, at round $r_0$ there exists a node in $C_{n/k}$ with buffer size at most $n/2+k-1\leq 2n/3$, implying that the number of disseminated messages in $C_{n/k}$ is at most $n/2+r_0+k-1\leq n/4$.
		At round $r_0$, at least $n/4$ messages are not disseminated in $C_{n/k}$ but are known to some nodes of the clique. Denote the set of these messages by $M_2$. Messages in $M_2$ were not  received from nodes within the clique $C_{n/k}$ (otherwise, there are disseminated), which means that at round $r_0$ every node in $C_{n/k}$ knows at most $r_0<n/\sqrt{k}$ such messages.
		During the $r_2$ rounds in the interval $[r_0, r_0+r_2]$, all buffers in all nodes in $C_{n/k}$ are of size at least $n/2-r_2\geq n/2-n/\sqrt{k}-n/24 \geq n/4$.
		Since buffer sizes are at least $n/4$, the probability $\hat{q}$ for each node in $C_{n/k}$ to send a message from $M_2$ in a single round is at most $\frac{n}{\sqrt{k}}/\frac{n}{4}=\frac{4}{\sqrt{k}}$.
		In order for the dissemination process to complete, each message in $M_2$ must be sent at least once by some node in $C_{n/k}$ (or be sent by all nodes of $C_{n/k-1}$, which happens only after $\Omega(n)$ rounds in expectation, by~\autoref{lemma:linear_buff}).
		By considering the sending of a message from $M_2$ a success, which occurs with probability at most $\hat{q}$, the dissemination process completes after at least $|M_2|\geq n/4$ successes. Denote by $X$ the random variable of number of trials before reaching $|M_2|$ successes. $X$ is a negative binomial variable, $X\sim NB(|M_2|,\hat{q})$, with expectation of $|M_2|/\hat{q}\geq \frac{n}{4}/\frac{4}{\sqrt{k}}= n\sqrt{k}$ trials.
		In each round, the number of trials is $k$ (one trial per node of the clique), and hence, the expected number $r_2$ of additional rounds before all messages are disseminated in $C_{n/k}$ is at least $n\sqrt{k}/k=\frac{n}{\sqrt{k}}$ in expectation. We get that
		$$	E\left[r_2 \ \middle|\ r_0<\frac{n}{\sqrt{k}}, r_2<\frac{n}{24}\right] \geq \frac{n}{\sqrt{k}} \enspace ,
		$$
		which proves~\eqref{eq:6}.
	\end{enumerate}
	
	In summary, we covered the whole space by combinations of events that form a partition, proved that the conditional expectation in each case is $\Omega(n/\sqrt{k})$, and hence by the law of total expectation, the theorem follows. \qed
\end{proof}

\subsection{Missing Proofs}
\label{sec:proofs}

\begin{figure}
	\centering
	\includegraphics[width=250pt]{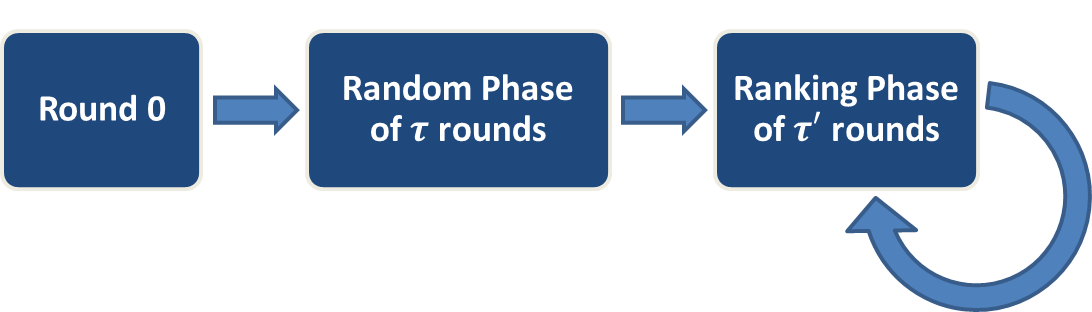}
	\caption{Phases of \autoref{alg:1}.}
	\label{fig:phases}
\end{figure}

\begin{lemma-repeat}{lemma:ranking}
	Let $m$ be a message with rank $r\leq 8 \tau$ (recall that $\tau=\alpha\log n$), then $m$ is sent during the ranking phase with probability at least $1-n^{-d}$.
\end{lemma-repeat}

\begin{proof}
	Let $A$ be the event that the message with rank $r$ is \emph{not} picked during a phase of $\tau^\prime=8 d \tau \log^{2} n=8 d \alpha \log^{3} n$ rounds.
	We wish to bound from above the probability for event $A$:
	\begin{equation*}
	\begin{split}
	\Pr[A] & \leq \left(1-  \frac{1}{ r \cdot H_b} \right)^{\tau^\prime}  \leq \left(1-  \frac{1}{ r \cdot (\ln b+1)} \right)^{\tau^\prime}  \leq \left(1-  \frac{1}{ r \cdot \log b} \right)^{\tau^\prime}  \leq \\
	& \leq \left(1- \frac{1}{r\log n}\right)^{8 d \alpha \log^{3} n} \leq \left(1- \frac{1}{r\log n}\right)^{(r\log n)\frac{1}{r}8 d \alpha\log^{2} n} \leq \\
	& \leq \left(\frac{1}{2}\right)^{\frac{1}{r}8 d \alpha\log^{2} n} = \left(\frac{1}{n}\right)^{\frac{1}{r}8 d \alpha\log n} \enspace .
	\end{split}
	\end{equation*}
	The second inequality holds because $H_n \leq \ln(n)+1$. The  last inequality holds since $(1-1/x)^x \leq e^{-1}<1/2$ for $x>0$.
	Namely, any message with $r\leq 8 \tau=8 \alpha\log n$ is sent during the phase with probability at least $1-\frac{1}{n^d}$.
	\qed
\end{proof}

\begin{lemma-repeat}{lemma:rand_pioneer}
	With probability at least $1-1/n^{\alpha/24-1}$,
at the end of the random phase, for every $i$, the number of pioneer messages that reach $C_i$ is $\leq 3 \tau$.
\end{lemma-repeat}
\begin{proof}
According to \emph{pioneer} definition, considering the direction of the flow of messages, cliques $C_1$ and $C_2$ could not have pioneer messages.
Fix $i$, $3\leq i \leq n/k$.
By \autoref{proposition:init_buff}, at the beginning of the random phase, for every node $u\in C_{i-1}$, buffer $\hat{B}_{u}(t_0)$ contains exactly one unique message $m_v, v\in C_{i-2}\cap L(u)$, and it holds that $|\hat{B}_u(t_0+t')|= k+1-t'$ during the random phase (as $C_{i-1}$ is an inner clique).

Let $\mathbbm{1}_{u}$, for every $u\in C_{i-1}$, be an indicator variable that indicates whether node $u$ sends its unique message during the random phase, or not.
Then
\begin{equation*}
\begin{split}
\Pr[\mathbbm{1}_{u}=1] & = 1-Pr[\mathbbm{1}_{u}=0] = 1-\prod_{t'=0}^{\tau-1} \frac{k-t'}{k+1-t'} = \\
& = 1-\frac{k+1-\tau}{k+1} = 1-\left(1-\frac{\tau}{k+1} \right) = \frac{\tau}{k+1} \enspace .
\end{split}
\end{equation*}
Let $X_{i-1}=\sum_{u\in C_{i-1}} \mathbbm{1}_{u}$, be the number of messages $m_v, v\in C_{i-2},$ that reach clique $C_i$ by the end of the random phase.
Then
\begin{equation*}
\begin{split}
\mu =E(X_{i-1}) &=E\left(\sum_{u\in C_{i-1}} \mathbbm{1}_{u}\right)=\sum_{u\in C_{i-1}} E(\mathbbm{1}_{u}) = \\
& = \sum_{u\in C_{i-1}} \frac{\tau}{k+1} = k\cdot\frac{\tau}{k+1} \enspace,
\end{split}
\end{equation*}
which means that $\tau/2 \leq \mu \leq \tau $.
The indicator variables are independent, as they refer to decisions of distinct nodes.
By applying a Chernoff bound, we get
\begin{equation*}
\begin{split}
\Pr[X_{i-1}>(3/2)\tau]  & \leq
\Pr[X_{i-1}\geq(3/2)\mu]\leq \Pr[X_{i-1}\geq(1+\delta)\mu]\leq \\
& \leq \exp\left(\frac{-\delta^2 \cdot \mu}{3}\right) \leq\exp\left(\frac{-\delta^2 \cdot (\tau/2)}{3}\right) \leq\\
&\leq \exp\left(\frac{-\delta^2 \cdot\alpha\log n}{6}\right) < \frac{1}{n^{\frac{\alpha\delta^2}{6}}} \enspace .
\end{split}
\end{equation*}
By setting $\delta=\frac{1}{2}$, we get that the number of pioneer messages, $X_{i-1}$, that reach~$C_i$ from one direction is~$\leq (3/2) \tau$ with probability at least $1- \frac{1}{n^{\alpha/24}}$.
By a union bound, this holds for both directions and every clique with probability at least $1-\frac{1}{n^{\alpha/24-1}}$.
\qed
\end{proof}

\begin{lemma-repeat}{lemma:h_send}
		With probability at least $1-\frac{1}{n^{d-2}}$,
	every node $u$ that starts ranking phase $p$ with at most $8 \tau$ fresh messages, sends all of them during the phase.
\end{lemma-repeat}
\begin{proof}
Fix a node $u$.
All fresh messages $m_v\in R_u(t)$ have rank $r \leq 8 \tau$.
According to \autoref{lemma:ranking}, a message with rank $r\leq 8 \tau$ is sent during a ranking phase with probability at least $1-\frac{1}{n^d}$.
By a union bound, the probability for node $u$ to send all of its fresh messages during the phase is bounded by $1-\frac{1}{n^d}\cdot 8 \tau\geq1-\frac{1}{n^{d-1}}$.
We use a union bound once more to bound the probability that this happens for every node $u$ by $1-\frac{1}{n^{d-1}}\cdot n=1-\frac{1}{n^{d-2}}$.
\qed
\end{proof}

\subsubsection{Property 1.}
For every $i, 1\leq i\leq\frac{n}{k}$, it holds that the number of messages $m_v$, $v\in C_{i-p-1}$, such that $m_v \in R_u(t_{p})$ for some $u\in C_i$ (pioneers), is at most $3 \tau$, and each reaches a distinct node $u\in L(v)$.

\subsubsection{Property 2.}
For every $i, 1\leq i\leq\frac{n}{k}$, and every node $u\in C_i$, it holds that at time $t_{p}$ there are at most $4 \tau$ \emph{fresh} messages $m_v$ for node $u$ for every one of the two directions of flow ($8 \tau$ in total).
All of them originated at nodes $v\in C_{i-p}$ (similarly, $v\in C_{i+p}$), except for at most one (a pioneer) which originated at $u'\in C_{i-p-1} \cap L(u)$ (similarly, $u'\in C_{i+p+1} \cap L(u)$).
All messages $m_v\in R_u(t_{p}), v\in C_{i-p}$ (similarly, $v\in C_{i+p}$), are fresh.

\subsubsection{Property 3.}
For every $i, 1\leq i\leq\frac{n}{k}$, and every node $v\in C_{i-p}$, it holds that $m_v$ is fresh for at least $T$ nodes $u\in C_i$ at time $t_{p}$.
Recall that $T=\frac{1}{2}\tau$.

\subsubsection{Property 4.}
For every $i, 1\leq i\leq\frac{n}{k}$, every node $u\in C_i$, and every node $v$ such that $v\in C_j$ for some $i-p\leq j \leq i$,
it holds that $m_v\in R_u(\bar{t}_p)$, and $m_v$ is non-fresh.

\paragraph{}
We prove the four properties simultaneously by induction on the ranking phase number, $p$.
To prove the base cases, we assume that all events described in \autoref{lemma:known_promotes}, \autoref{lemma:rand_pioneer}, and \autoref{lemma:h_send} (for $p=1$) occur. Notice that, by a union bound, the probability for this is at least $1-\left(\frac{1}{n^{\alpha/24-1}}+\frac{1}{n^{\alpha/48-1}}+\frac{1}{n^{d-2}}\right)\geq 1-\left(\frac{2}{n^{\alpha/48-1}}+\frac{2}{n^{d-2}}\right)$.

\begin{proof}[Base case for Property 1]
Let $p=1$.
One random phase precedes the first ranking phase.
The upper bound on the number of pioneers in every clique holds according to \autoref{lemma:rand_pioneer}.
The distribution among distinct layers is immediate according to Attribute~\ref{pio:layer} of pioneer messages.
\qed
\end{proof}

\begin{proof}[Base case for Property 2]
Let $p=1$.
Fix some node $u\in C_i$.
We analyze possibilities for fresh messages for one direction of flow at the end of the random phase, and the other direction is symmetric.
By \autoref{proposition:p_distance}, messages $m_v\in R_u(t_1)$ originate at nodes $v\in {C_{i-2} \cup C_{i-1} \cup C_{i}}$.
\begin{inlinelistroman}
\item[] A message $m_v\in R_u(t_1)$ that originates at node $v\in C_{i-2}$ is a pioneer.
By Attributes~\ref{pio:layer} and~\ref{pio:cnt1} there can be at most one such message, and it is fresh.
\item[] For messages $m_v\in R_u(t_1)$ that originate at nodes $v\in C_{i-1}$ there are two possibilities.
One possibility is that they are received from the neighbor $u' \in C_i \cap L(v)$, which implies that they are pioneers in nodes $u_1\in C_{i+1}\cap L(v)$ at time $t_1$.
By Property~1 for $p=1$ (which is already proved),
there are at most $3 \tau$ such messages.
The only other possibility is that they are received from the neighbor $u' \in C_{i-1} \cap L(u)$. There are at most $\tau$ such messages (which might include one that originates at $C_{i-2}$, as already discussed), and they are all fresh.
\item[] Messages $m_v\in R_u(t_1)$ that originate at nodes $v\in C_{i}$ are all non-fresh, according to \autoref{lemma:known_promotes}.
\end{inlinelistroman}

In total,
at the beginning of the first ranking phase, each node $u$ has at most $4 \tau$ fresh messages from the one direction.
All of them originated at nodes $u'\in C_{i-1}$, except for at most one which originated at $u'\in C_{i-2} \cap L(u)$.
All messages that originated at nodes $u'\in C_{i-1}$ are fresh.
The other direction of flow is symmetric.
\qed
\end{proof}

\begin{proof}[Base case for Property 3]
Let $p=1$.
For every $v\in C_{i-1}$, at the end of round 0, exactly one node $u\in C_i$ knows $m_v$. It may disseminate it during the random phase.
At the end of the random phase, by \autoref{lemma:known_promotes},
for every $v\in C_{i-1}, m_v$ is non-fresh in all nodes of $C_{i-1}$. That is, by the end of the random phase, every node $v'\in C_{i-1}, v'\neq v,$ receives $m_v$ at least $T$ times,
all from nodes within the clique.
Therefore, at least $T$ nodes in $C_{i-1}$ send $m_v$ in the random phase, which implies that at least $T$ nodes in $C_i$ know $m_v$.
According to the phase separation property, every such node in $C_i$ receives $m_v$ at most twice (from the neighbor in $C_{i-1}$, and possibly from the neighbor $u\in C_i$), so it is fresh.
\qed
\end{proof}

\begin{proof}[Base case for Property 4]
Let $p=1$.
Fix $i$, $u\in C_i$.
According to \autoref{lemma:known_promotes},
it holds that for every node $v\in C_i$, $m_v$ is known and non-fresh in $u$.

At the beginning of the first ranking phase, according to Property 3 for $p=1$ and $i$, it holds
that every message $m_v, v\in C_{i-1}$, is fresh in at least $T$ nodes in $C_i$.
According to Property 2 for $p=1$, it holds
that every node has at most $8\tau$ fresh messages.
By \autoref{lemma:h_send},
all nodes (in particular, nodes in $C_i$) send all of their fresh messages.
This means that every message $m_v, v\in C_{i-1}$, is received by node $u$ at least $T$ times so it becomes non-fresh.
\qed
\end{proof}
This completes the proof of the base cases.
Recall that the base cases are proved by assuming that all events described in \autoref{lemma:known_promotes}, \autoref{lemma:rand_pioneer}, and \autoref{lemma:h_send} (for $p=1$) occur.
Thus, the properties are proved for $p=1$ with probability at least $1-\left(\frac{2}{n^{\alpha/48-1}}+\frac{2}{n^{d-2}}\right)$.

To prove the induction step, we assume that all events described in the four properties for $p-1$, and in \autoref{lemma:h_send} for $p-1$ and $p$, occur.
This happens with probability at least $1-\left(\frac{2}{n^{\alpha/48-1}}+\frac{2(p-1)}{n^{d-2}}+\frac{1}{n^{d-2}}+\frac{1}{n^{d-2}}\right)=1-\left(\frac{2}{n^{\alpha/48-1}}+\frac{2p}{n^{d-2}}\right)$.
\begin{proof}[Induction step for Property 1]
By Property 1 for $p-1$ and $i-1$,
at the beginning of ranking phase $p-1$, the number of messages $m_v$, $v\in C_{i-p-1}$, that reach nodes in $C_{i-1}$ is at most $3 \tau$, each reaches a distinct node $u\in C_{i-1} \cap L(v)$.
At time $t_{p-1}$, by Pioneer Attribute~\ref{pio:cnt1}, each one of them is fresh.
By Property 2 for $p-1$ and $i-1$,
at the beginning of ranking phase $p-1$, every node $u\in C_{i-1}$ has at most $8 \tau$ \emph{fresh} messages.
By \autoref{lemma:h_send} for $p-1$,
every node sends all of its fresh messages during ranking phase $p-1$ (in particular, pioneer messages in nodes in $C_{i-1}$).
Thus,
it holds that the number of messages $m_v$, $v\in C_{i-p-1}$, such that $m_v$ is a pioneer at time $t_p$ in nodes of $C_i$, is at most $3 \tau$, and each reaches a distinct node $u\in C_{i} \cap L(v)$.
\qed
\end{proof}

\begin{proof}[Induction step for Property 2]
Fix a node $u\in C_i$.
By \autoref{proposition:p_distance}, for every message $m_v\in R_u(t_p)$ (known to $u$ at the beginning of ranking phase $p$) it holds that $v\in{\bigcup \limits _{j \in \{i-p-1, \ldots, i\}} C_j}$.
By Property 4 for $p-1$ and $i$,
for every node $v$ such that $v\in C_j$ for some $i-p+1\leq j \leq i$,
it holds that $m_v\in R_u(\bar{t}_{p-1})$, $m_v$ non-fresh.
Thus, only messages $m_v, v\in C_{i-p-1} \cup C_{i-p}$ can be fresh.

Consider a message $m_v, v\in C_{i-p}$:
By Property 1 for $p-1$ and $i$,
at the beginning of ranking phase $p-1$, any message $m_v, v\in C_{i-p}$, that reach $C_i$ (a pioneer) is known to exactly one node in the clique.
Thus, any message $m_v, v\in C_{i-p}$, that reaches $C_i$ by the beginning of ranking phase $p$ is fresh (because it could be received only once from a neighbor within the clique $C_i$ and once from a neighbor in clique $C_{i-1}$, i.e., it is received at most twice).

By Property 2 for $p-1$ and $i-1$,
at the beginning of ranking phase $p-1$, node $u'\in C_{i-1} \cap L(u)$ has at most $4 \tau$ \emph{fresh} messages (consider relevant direction of flow),
all of them originated at nodes $v\in C_{i-p}$, except for at most one which originated at $u'\in C_{i-p-1} \cap L(u)$ (a pioneer).
According to \autoref{lemma:h_send} for $p-1$,
every node $u'$ sends all of its fresh messages during ranking phase $p-1$.
Thus, at the end of ranking phase $p-1$ (beginning of ranking phase $p$), they all reach $u$,
and they are all fresh. In particular, they are received at most twice, according to the previous discussion.
The opposite direction of flow is symmetric.
This completes the proof.
\qed
\end{proof}

\begin{proof}[Induction step for Property 3]
By Property 3 for $p-1$ and $i-1$,
at the beginning of ranking phase $p-1$, every message $m_v$, $v\in C_{i-p}$, is fresh in at least $T$ nodes $u'\in C_{i-1}$.
By Property 2 for $p-1$ and $i-1$,
at the beginning of ranking phase $p-1$, every node $u\in C_{i-1}$ has at most $8 \tau$ \emph{fresh} messages.
According to \autoref{lemma:h_send} for $p-1$,
all are sent during ranking phase $p-1$, each of the nodes $u\in C_{i-1}$ sends to a distinct neighbor node $u\in C_i$.
Therefore, at the end of ranking phase $p-1$ (beginning of ranking phase $p$), every message $m_v$, $v\in C_{i-p}$, is known to at least $T$ nodes $u\in C_i$.
By Property 2 for $p$ (which is already proved) and $i$,
all messages $m_v$, $v\in C_{i-p}$ known in $C_i$ are fresh,
which completes the proof.
\qed
\end{proof}

\begin{proof}[Induction step for Property 4]
By Property 4 for $p-1$ and $i$,
for every node $u\in C_i$, every node $v$ such that $v\in C_j$ for some $i-p+1\leq j \leq i$,
it holds that $m_v\in R_u(\bar{t}_{p-1})$, and $m_v$ is non-fresh.
This holds also at the end of ranking phase $p$.
We still need to show that the property holds for all message $m_v$, $v\in C_{i-p}$.
Notice that Properties 1,2 and 3 are already proved for $p$.

By Property 3 for $p$ and $i$,
at the beginning of ranking phase $p$, every message $m_v$, $v\in C_{i-p}$, is fresh in at least $T$ nodes $u\in C_i$.
By Property 2 for $p$ and $i$,
at the beginning of ranking phase $p$, every node $u\in C_i$ has at most $8 \tau$ fresh messages.
By \autoref{lemma:h_send} for $p$,
all are sent during ranking phase $p$.
This means that every message $m_v$, $v\in C_{i-p}$, is sent by at least $T$ nodes of the clique $C_i$.
This implies that every message $m_v, v\in C_{i-p}$, is received by every node $u\in C_i$ at least $T$ times.
Thus, at the end of ranking phase $p$, every message $m_v$, $v\in C_{i-p}$ is known and non-fresh in all nodes $u\in C_i$,
which completes the proof.
\qed
\end{proof}

Property 4 guarantees that full information spreading is completed after ranking phase $p=n/k$, with probability at least $1-\left(\frac{2n/k}{n^{d-2}}+\frac{2}{n^{\alpha/48-1}}\right)\geq 1-\left(\frac{1}{n^{d-3}}+\frac{1}{n^{\alpha/48-2}}\right) \geq 1-\frac{1}{n^c}$, for a constant $c$, by fixing $d$ and $\alpha$ to values $d>c+3, \alpha >48c+96$.
This completes the proof of \autoref{lemma:iteration}, from which \autoref{th:main1} follows.

\begin{theorem-repeat}{thm:alg2_time}
		\autoref{alg:shuffle} completes full information spreading on $G_{n,k}$ in $O\left(\frac{n}{k}\log^3 n\right)$ rounds, w.h.p.
\end{theorem-repeat}

Proving the four properties for the modified algorithm implies \autoref{lemma:iteration}, from which \autoref{thm:main2} follows.
In the previous analysis, the transition from the end of a ranking phase to the beginning of the next one was immediate, therefore claims that hold at end of ranking phase $p-1$, automatically hold at the beginning of ranking phase $p$.
Here, every two consecutive ranking phases are separated by a shuffle phase,
implying that $\bar{t}_{p-1}$ and $t_{p}$ are not equal anymore.
We need to prove that the relevant claims that hold at the beginning of a shuffle phase (end of a ranking phase) hold also at the end of the shuffle phase (beginning of the next ranking phase).
That is, we prove that shuffle phases preserve the required properties.
The addition of the shuffle phase does not affect the progress of the algorithm until the end of the first ranking phase.
Thus, the base case in the inductive proof of the four properties stays as is.
Modifications are needed to the proofs of inductive steps.

Before heading to modify the proof of the induction step, we first prove the following.
\begin{lemma}
	\label{lemma:k_per_side}
	Assume properties 2, 3 and 4 hold at the end of ranking phase $p-1$.
	Then, for each node $u\in C_i$,
	and for each direction of flow, at the end of shuffle phase~$p-1$, there are $k$ remaining messages $m_v, v\in C_{i-p}$ (similarly $C_{i+p}$) in $R$ after filtering out unwanted messages (in line~\ref{alg2:lineR}).
\end{lemma}
\begin{proof}
	Properties 2 and 3 hold at the end of ranking phase $p-1$,
	i.e., at the beginning of shuffle phase $p-1$, for every node $v\in C_{i-p}$, it holds that $m_v$ is fresh in at least $T$ nodes $u\in C_i$,
	and that every node in $C_i$ has at most $4\tau$ fresh messages per direction.
	Thus, during the shuffle phase, every message $m_v, v\in C_{i-p}$, is sent (and thus, received) at least $T$ times by nodes of $C_i$, and therefore is not filtered out at the end of the shuffle phase.
	As already discussed, messages that originate at $m_v, v\in C_{i-p-1}$ are filtered out due to low $phasecnt$ values.
	By property 4 for end of ranking phase $p-1$,
	for every node $v$ such that $v\in C_j$ for some $i-p+1\leq j \leq i$,
	it holds that $m_v$ non-fresh, so they are filtered out.
	In total, all messages $m_v, v\in C_{i-p}$, are not filtered out, and only them.
	The other direction of flow is symmetric.
	\qed
\end{proof}

\begin{lemma}
	\label{lemma:shuffle_threshold}
	Assume properties 2, 3 and 4 hold at the end of ranking phase $p-1$.
	Then,
	with probability at least $1-\frac{1}{n^{9\alpha/16-1}}$,
	at the end of shuffle phase $p-1$,
	every message that is not filtered out in node $u\in C_i$, is selected to be fresh by at least $T$ nodes in $C_i$.
\end{lemma}

\begin{proof}
	Assume properties 2, 3 and 4 hold at the end of ranking phase $p-1$.
	Fix $i, v$.
	Let $\mathbbm{1}_{u,v}$, for every $u\in C_i$, be indicator variables that indicate whether node $u$ selects $m_v$ at the end of shuffle phase $p-1$, or not.
	By \autoref{lemma:k_per_side}, there are at most $2k$ remaining messages in $R$ after filtering out unwanted messages (in line~\ref{alg2:lineR}).
	Thus, the probability for each message to be within the $4\tau$ selected messages at the end of the shuffle phase is at least
	\begin{equation*}
	\begin{split}
	\Pr[\mathbbm{1}_{u,v}=1] & \geq  \frac{4\tau}{2k} = \frac{2\tau}{k} \enspace .
	\end{split}
	\end{equation*}
	Let $X_{v}=\sum_{u\in C_i} \mathbbm{1}_{u,v}$, be the number of nodes in $C_i$ that select message $m_v$ at the end of shuffle phase $p-1$.
	Then
	\begin{equation*}
	\begin{split}
	\mu =E(X_{v}) &=E\left(\sum_{u\in C_i} \mathbbm{1}_{u,v}\right)=\sum_{u\in C_i} E(\mathbbm{1}_{u,v}) \geq \\
	& \geq \sum_{u\in C_i} \frac{2\tau}{k} = k\cdot\frac{2\tau}{k} = 2\tau \enspace .
	\end{split}
	\end{equation*}
	The indicator variables are independent, as they refer to decisions of distinct nodes.
	By applying a Chernoff bound, we get
	\begin{equation*}
	\begin{split}
	\Pr[X_{v}\leq(1-\delta)\mu] & \leq \exp\left(-\delta^2\frac{\mu}{2}\right) \leq \exp\left(-\delta^2\frac{2\tau}{2}\right) = \\
	& = \exp\left(-\delta^2\alpha\log n\right) < \frac{1}{n^{\alpha\delta^2}} \enspace .
	\end{split}
	\end{equation*}
	By setting $\delta=\frac{3}{4}$, we get that a message $m_v$ is selected fresh in at least $T$ nodes $u\in C_i$ with probability at least $1- \frac{1}{n^{9\alpha/16}}$.
	By a union bound, this holds for every node $v$ with probability at least $1-\frac{1}{n^{9\alpha/16-1}}$.
	\qed
\end{proof}

To match the modification of the algorithm, we show that the four properties now hold for $p$ with probability at least $1-\left(\frac{2}{n^{\alpha/48-1}}+\frac{2p}{n^{d-2}}+\frac{p}{n^{9\alpha/16-1}}\right)$.
To prove the new induction step, we make similar assumptions as earlier when proving the induction step, i.e., all events described in the four properties for $p-1$, and in \autoref{lemma:h_send} for $p-1$ and $p$, occur.
In addition, we assume that events described in \autoref{lemma:shuffle_threshold} for $p-1$, occur.
In total, this happens with probability at least
\begin{equation*}
\begin{split}
&1-\left(\frac{2}{n^{\alpha/48-1}}+\frac{2(p-1)}{n^{d-2}}+\frac{p}{n^{9\alpha/16-1}}+\frac{1}{n^{d-2}}+\frac{1}{n^{d-2}}\right)= \\
&1-\left(\frac{2}{n^{\alpha/48-1}}+\frac{2p}{n^{d-2}}+\frac{p}{n^{9\alpha/16-1}}\right) \enspace .
\end{split}
\end{equation*}

\begin{proof}[Extension of induction step for property 1]
	The property holds at the end of ranking phase $p-1$.
	At the beginning of shuffle phase $p-1$, each pioneer message in a clique is known to exactly one node in the clique.
	Thus, at the end of the shuffle phase, the $phasecnt$ values for pioneer messages are at most 2 (one reception is from the respective node within the same clique, and the other is from the neighbor from the neighboring clique).
	In conclusion, all pioneer messages are filtered out, so there are no pioneer messages at the beginning of ranking phase $p$, which completes the proof.
	\qed
\end{proof}

\begin{proof}[Extension of induction step for property 2]
	The property holds at the end of ranking phase $p-1$.
	At the beginning of shuffle phase $p-1$,
	considering one direction of flow,
	all fresh messages $m_v$ in nodes of clique $C_i$
	originate at nodes $C_{i-p}$, except for pioneers (originating at nodes in $C_{i-p-1}$).
	At the end of shuffle phase $p-1$, as already discussed, all pioneer messages are filtered out due to low $phasecnt$ values.
	By property 4 for the end of ranking phase $p-1$,
	all messages $m_v\notin C_{i-p}$ are non-fresh, so they are filtered out (if any) for being non-fresh prior to the start of shuffle phase $p-1$.
	Thus, in total, considering both directions, at the end of shuffle phase $p-1$, each node selects $4 \tau$ of the messages $m_v, v\in C_{i-p} \cup C_{i+p}$, marks them fresh and ranks them $1$ to $4\tau$.
	This completes the proof.
	\qed
\end{proof}

\begin{proof}[Extension of induction step for property 3]
	Properties 2 and 3 hold at the end of ranking phase $p-1$,
	i.e., at the beginning of shuffle phase $p-1$, for every node $v\in C_{i-p}$, it holds that $m_v$ is fresh in at least $T$ nodes $u\in C_i$,
	and that every node in $C_i$ has at most $4\tau$ fresh messages per direction.
	During the shuffle phase, every message $m_v, v\in C_{i-p},$ is sent at least $T$ times by nodes of $C_i$, and therefore is not filtered out at the end of the shuffle phase.
	By \autoref{lemma:shuffle_threshold} for $p-1$, each message is selected and becomes fresh in at least $T$ nodes,
	which completes the proof.
	\qed
\end{proof}

\begin{proof}[Extension of induction step for property 4]
	The original proof of property 4 for $p$ shown in the previous section relies on
	property 4 at the end of ranking phase $p-1$,
	on Properties 2 and 3 at the beginning of ranking phase $p$, and on \autoref{lemma:h_send} for $p$.
	At this point, all of them are proved.
	Thus, the same original proof for property 4 applies directly.
	
	In other words,
	by property 4 for $p-1$ and $i$,
	for every node $u\in C_i$, every node $v$ such that $v\in C_j$ for some $i-p+1\leq j \leq i$,
	it holds that $m_v\in R_u(\bar{t}_{p-1})$, and $m_v$ is non-fresh.
	Notice that shuffle phases preserve this.
	By property 3 for $p$ and $i$,
	at the beginning of ranking phase $p$, every message $m_v$, $v\in C_{i-p}$, is fresh in at least $T$ nodes $u\in C_i$.
	By property 2 for $p$ and $i$,
	at the beginning of ranking phase $p$, every node $u\in C_i$ has at most $8 \tau$ fresh messages.
	By \autoref{lemma:h_send} for $p$,
	all are sent during ranking phase $p$.
	This means that every message $m_v$, $v\in C_{i-p}$, is sent by at least $T$ nodes of the clique $C_i$.
	This implies that every message $m_v, v\in C_{i-p}$, is received by every node $u\in C_i$ at least $T$ times.
	Thus, at the end of ranking phase $p$, every message $m_v$, $v\in C_{i-p}$ is known and non-fresh in all nodes $u\in C_i$,
	which completes the proof.
	\qed
\end{proof}

This completes the proof.
Recall that we assumed that all events described in the four properties for $p-1$, in \autoref{lemma:h_send} for $p-1$ and $p$, and in \autoref{lemma:shuffle_threshold} for $p-1$, occur.
Thus, the properties are proved with probability at least $1-\left(\frac{2}{n^{\alpha/48-1}}+\frac{2p}{n^{d-2}}+\frac{p}{n^{9\alpha/16-1}}\right)$.

Assigning $p=n/k$ in Property 4 proves \autoref{lemma:iteration}, from which \autoref{thm:main2} follows, with
probability at least
\begin{equation*}
\begin{split}
&1-\left(\frac{2}{n^{\alpha/48-1}}+\frac{2n/k}{n^{d-2}}+\frac{n/k}{n^{9\alpha/16-1}}\right)\geq  \\
&1-\left(\frac{1}{n^{d-3}}+\frac{1}{n^{\alpha/48-2}}+\frac{1}{n^{9\alpha/16-2}}\right) \geq 1-\frac{1}{n^c} \enspace ,
\end{split}
\end{equation*}
for a constant $c$, by fixing $d$ and $\alpha$ to values $d>c+3, \alpha >48c+96$.

\begin{figure}
	\centering
	\includegraphics[width=300pt]{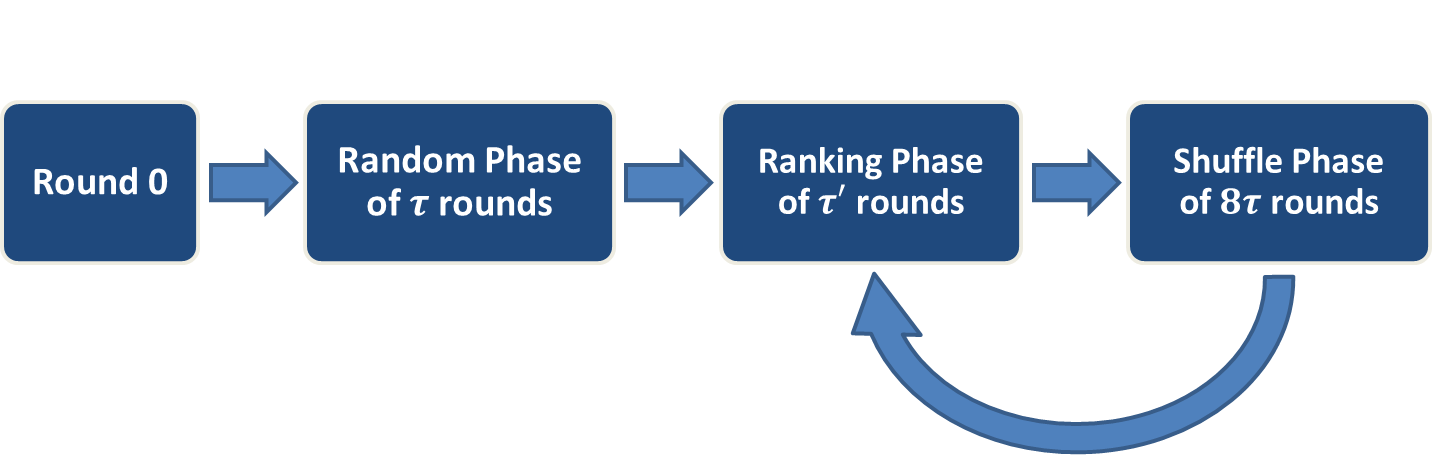}
	\caption{Phases of \autoref{alg:shuffle}.}
	\label{fig:phases_2}
\end{figure}

\begin{lemma-repeat}{lemma:nonfaulty_clique}
		At the end of round ${\tau_{e}}$, the number of non-faulty nodes in each clique is at least $(30k/32)$, with probability at least $1-1/n^{30}$.
\end{lemma-repeat}
\begin{proof}

Let $\mathbbm{1}_{u}$, for every node $u$, be an indicator variable that indicates whether node $u$ is non-faulty after ${\tau_{e}}$ rounds, or not.
Then
\begin{equation*}
\begin{split}
\Pr[\mathbbm{1}_{u}=1] & = (1-q)^{\tau_{e}}\geq 1-q{\tau_{e}}\geq 1-1/32=31/32 \enspace .
\end{split}
\end{equation*}
Let $X_{i}=\sum_{u\in C_i} \mathbbm{1}_{u}$, for every $i, 1\leq i\leq n/k$, be the number of non-faulty nodes in $C_i$ after ${\tau_{e}}$ rounds.
Then
\begin{equation*}
\begin{split}
\mu =E(X_{i}) &=E\left(\sum_{u\in C_i} \mathbbm{1}_{u}\right)=\sum_{u\in C_i} E(\mathbbm{1}_{u}) \geq \\
& \geq \sum_{u\in C_i} 31/32 = 31k/32 \enspace .
\end{split}
\end{equation*}
The indicator variables are independent, as failure events of nodes are independent.
By applying a Chernoff bound, with $\delta=\frac{1}{31}$, we get
\begin{equation*}
\begin{split}
\Pr\left[X_{i}<\frac{30}{32}k\right] & \leq \Pr\left[X_{i}\leq(1-\delta)\frac{31}{32}k \right] \leq \Pr[X_{i}\leq(1-\delta)\mu] \leq \exp\left(-\delta^2\frac{\mu}{2}\right) \leq \\
& \leq \exp\left(-\delta^2\frac{31k}{2\cdot32}\right) \leq \exp\left(-\delta^2\frac{31(2\cdot32\cdot31^2\log n)}{2\cdot32}\right) < \\
& < \frac{1}{n^{31^3\delta^2}} \enspace .
\end{split}
\end{equation*}
The inequality in second line holds because $k=\Omega(\log^3n)$.
We get that at the end of round ${\tau_{e}}$, the number of non-faulty nodes in a clique is at least $(30k/32)$ with probability at least $1- \frac{1}{n^{31}}$.
By a union bound, this holds for every clique with probability at least $1-\frac{1}{n^{30}}$.
\qed
\end{proof}

\begin{theorem-repeat}{thm:main2}
	\autoref{alg:shuffle} completes full information spreading on $G_{n,k}$ in $O\left(\frac{n}{k}\log^3 n\right)$ rounds, for any node failure probability per round $q$, $0\leq q \leq O\left(\frac{k}{n \log^3 n}\right)$, w.h.p.
\end{theorem-repeat}
\begin{proof}
	Fix $i, p$.
	Let $m_v$ be a message that is fresh in at least $T$ (non-faulty) nodes in $C_{i-1}$ at the end of shuffle phase $p-1$.
	Here we analyze the probability
	that $m_v$ is \emph{not} shuffled successfully in clique $C_i$.
	An unsuccessful shuffle might occur either because the $phasecnt$ values in $C_i$ at the end of shuffle phase $p$ are smaller than the threshold of $T^{*}=\hat{c}T$, so the message is filtered out (denote this event by $A$), or because the message was selected by less than $T$ (non-faulty) nodes.
	By~\autoref{lemma:h_send}, at the beginning of shuffle phase $p$, the message $m_v$ is supposed to be fresh in at least $T$ nodes in $C_i$ (each of them gets the message from its respective neighbor in $C_{i-1}$).
	Of these nodes in $C_i$, if one does not send $m_v$ during shuffle phase $p$, then either the node or its neighbor in $C_{i-1}$ (or both) becomes faulty by the end of shuffle phase~$p$.
	The probability, $\hat{q}$, for such a pair of nodes \emph{not} to fail is bounded from below (according to Bernoulli's inequality) by $\hat{q} = ((1-q)^{{\tau_{e}}})^2 \geq (1-q{\tau_{e}})^2\geq 1-2q{\tau_{e}}\geq1-1/16$.
	Fix a set of $T$ pairs of nodes $S(m_v)\subseteq C_{i-1} \times C_i$, of those who know message $m_v$ in $C_{i-1}$ at the end of shuffle phase $p-1$, and their respective neighbors in $C_i$.
	There might exist more than $T$ such pairs, but by fixing a set of size $T$ and ignoring the rest, we bound the probability of an unsuccessful shuffle from above, as the ignored nodes can only help and increase the probability of success.
	A ``surviving'' pair is a pair of nodes from $S(m_v)$ where both are non-faulty at the end of the shuffle phase, and hence function properly (by sending message $m_v$) during shuffle phase $p$.
	Denote by $s$, the number of ``surviving'' pairs.
	We have that
	\begin{equation*}
	\begin{split}
	\Pr[A] & \leq \sum\limits_{s=0}^{T^{*}-1} {{T \choose s} \cdot(\hat{q})^s \cdot (1-\hat{q})^{T-s}} \leq \sum\limits_{s=0}^{T^{*}-1} {{T \choose s} \cdot (1-\hat{q})^{T-s}} \leq \\
	& \leq \sum\limits_{s=0}^{T^{*}-1} {{T \choose s} \cdot (1/16)^{T-s}} \enspace .
	\end{split}
	\end{equation*}
	We sum over all $s\in\{0 , \ldots, T^{*}-1 \}$, where the number of ``survivors'' is lower than the threshold of~$\hat{c}T$, which implies that the message $m_v$ is filtered out, improperly, at the end of the shuffle phase due to a low $phasecnt$ value.
	By setting $0 < \hat{c}\leq\frac{1}{2}$, we get the following,
	\begin{equation*}
	\begin{split}
	\Pr[A] & \leq T^{*} \cdot {T \choose T/2 } \cdot  (1/16)^{T/2} \leq \hat{c}T \cdot \left(\frac{T\cdot e}{T/2}\right)^{T/2} \cdot (1/16)^{T/2} \leq \\
	& \leq T/2 \cdot \left((2e)^\frac{1}{2}\right)^{T}\cdot (1/16)^{T/2} \leq {\frac{1}{4}\alpha\log n} \cdot \left((2e)^\frac{1}{4}\right)^{\alpha\log n}\cdot \left(\frac{1}{2^{4}}\right)^{\frac{1}{4}\alpha\log n} \leq \\
	& \leq n \cdot \left(2^\frac{2}{3}\right)^{\alpha\log n}\cdot \left(\frac{1}{2^{4}}\right)^{\frac{1}{4}\alpha\log n} \leq n \cdot n^{\frac{2}{3}\alpha}\cdot \left(\frac{1}{n^{4}}\right)^{\frac{1}{4}\alpha} \leq n^{\frac{2}{3}\alpha+1}\cdot \frac{1}{n^{\frac{4}{4}\alpha}} = \\
	& = \frac{1}{n^{\alpha-\frac{2}{3}\alpha-1}} = \frac{1}{n^{\alpha/3-1}} \enspace .
	\end{split}
	\end{equation*}
	Namely, the message is not filtered out with probability at least $\frac{1}{n^{\alpha/3-1}}$.
	The number of non-faulty nodes in each clique is at least $31k/32$ with probability at least $1-\frac{1}{n^{30}}$, by~\autoref{lemma:nonfaulty_clique}.
	An analysis similar to the one in the proof of~\autoref{lemma:shuffle_threshold} (with $\delta=11/15$) gives that,
	once the message is not filtered out, it is selected by at least $T$ of the non-faulty nodes in $C_i$ with probability at least $1-\frac{1}{n^{11^2\alpha/(15\cdot16)}}$.
	In total, by using a union bound, a message is not shuffled successfully between two consecutive shuffle phases with probability at most $\frac{1}{n^{\alpha/3-1}}+\frac{1}{n^{11^2\alpha/(15\cdot16)}}+\frac{1}{n^{30}} \leq \frac{1}{n^{6}}$ (for value of $\alpha$ fixed earlier).
	
	We use union bound two more times, for all messages and for all phases, and get an upper bound for the probability that a message is not propagated properly, of $\frac{1}{n^{4}}$.
	This proves that the algorithm tolerates failures that occur with probability $0\leq q\leq \frac{1}{32{\tau_{e}}}$ in the given model, with probability at least $1-\frac{1}{n^{4}}$.
	\qed
\end{proof}

\end{document}